\newif\ifTwoColumn
\newif\ifTechReport

\pdfoutput=1

\ifTwoColumn
	\documentclass[9pt,twocolumn]{IEEEtran}
	
\else
	\ifTechReport
		\documentclass[11pt,onecolumn]{IEEEtran}
		\usepackage[width=6.25in]{geometry}
		\usepackage{parskip}

	\else
		\documentclass[12pt,draftclsnofoot,onecolumn,letterpaper]{IEEEtran}
	
	\fi 
\fi

\IEEEoverridecommandlockouts                              
\usepackage{amssymb}  
\usepackage{graphics} 
\usepackage[noadjust]{cite}
\usepackage{psfrag, color}
\usepackage[final]{showkeys}
\usepackage{empheq}
\usepackage{upgreek}
\usepackage[ruled]{algorithm}
\usepackage{algorithmic}
\usepackage{amsmath}
\usepackage{amsfonts}
\usepackage{acronym}
\usepackage{epsfig}
\usepackage{pifont}
\usepackage[]{xcolor}
\usepackage{caption,subcaption}
\usepackage{flushend} 
\usepackage{enumerate}

\definecolor{LinkColor}{rgb}{1,0,0}	   

\graphicspath{{figs/}}

\acrodef{SPD}[SPD]{symmetric positive definite}
\acrodef{SPSD}[SPSD]{symmetric positive semi-definite}
\acrodef{QP}[QP]{quadratic program}
\acrodef{QPs}[QPs]{quadratic programs}
\acrodef{DSPs}[DSPs]{digital signal processors}
\acrodef{PLCs}[PLCs]{programmable logic controllers}
\acrodef{MPC}[MPC]{Model predictive control}
\acrodef{FPGAs}[FPGAs]{field programmable gate arrays}
\acrodef{FPGA}[FPGA]{field programmable gate array}
\acrodef{AFM}[AFM]{atomic force microscope}
\acrodef{LTI}[LTI]{linear time-invariant}

\newtheorem{remark}{Remark}
\newtheorem{assumption}{Assumption}
\newtheorem{proposition}{Proposition}
\newtheorem{theorem}{Theorem}
\newtheorem{lemma}{Lemma}

\usepackage{calc}

\newcommand{\diag}[1]{\operatorname{diag}\!\left( #1 \right)}
\newcommand{\SetK}{\mathbb{K}}
\newcommand{\SetA}{\mathbb{A}}
\newcommand{\SetU}{\mathbb{U}}
\newcommand{\SetX}{\mathbb{X}}
\newcommand{\SetQ}{\mathbb{Q}}
\newcommand{\SetXD}{{\mathbb{X}_{\Delta}}}
\newcommand{\Reals}[1]{\mathbb{R}^{#1}}
\newcommand{\Hm}{C}
\newcommand{\bm}{\gamma}
\newcommand{\Am}{\mathcal{M}}

\newcommand{\low}{\min}
\newcommand{\upp}{\max}
\newcommand{\free}{\text{$\mathcal{F}$}}
\newcommand{\bound}{\text{$\mathcal{B}$}}
\newcommand{\soft}{\text{$\mathcal{S}$}}
\newcommand{\cent}{\text{c}}

\newcommand{\set}[2]{\left\{ #1\ \left| \ #2 \right. \right\}}
\newcommand{\eqdef}{:=}
\newcommand{\norm}[1]{\left\| #1 \right\|}
\newcommand{\half}{\frac{1}{2}}
\DeclareMathOperator{\image}{Im}

\DeclareMathOperator{\kernel}{Ker}

\title{\LARGE Embedded Online Optimization for Model Predictive Control at Megahertz Rates}

\author{Juan L.~Jerez, Paul J.~Goulart, Stefan Richter, George A.~Constantinides, \\Eric C.~Kerrigan and Manfred Morari
\thanks{Juan L.~Jerez and George A.~Constantinides are with the Department of Electrical and Electronic Engineering, Imperial College London,
        SW7 2AZ, United Kingdom, {\tt\small jlj05|gac1@imperial.ac.uk}}%
\thanks{Paul J.~Goulart, Stefan Richter and Manfred Morari are with the Automatic Control Laboratory, ETH Z\"urich, 8092 Z\"urich, Switzerland,
        {\tt\small pgoulart|richters|morari@control.ee.ethz.ch}}%
\thanks{Eric C.~Kerrigan is with the Department of Electrical and Electronic Engineering and the Department of Aeronautics, Imperial College London,
        SW7 2AZ, United Kingdom,
        {\tt\small e.kerrigan@imperial.ac.uk}}%
}

\begin{document}

\maketitle




\begin{abstract}
%
%
\ifTwoColumn
\else
	\ifTechReport
	\else
		\vspace{-1ex}
	\fi
\fi

Faster, cheaper, and more power efficient optimization solvers than those currently offered by general-purpose solutions are required for extending the use of model predictive control (MPC) to resource-constrained embedded platforms. We propose several custom computational architectures for different first-order optimization methods that can handle linear-quadratic MPC problems with input, input-rate, and soft state constraints. We  provide analysis ensuring the reliable operation of the resulting controller under reduced precision fixed-point arithmetic. Implementation of the proposed architectures in FPGAs shows that satisfactory control performance at a sample rate beyond 1~MHz is achievable even on low-end devices, opening up new possibilities for the application of MPC on embedded systems. 

\end{abstract}


\section{Introduction}

\ac{MPC} provides a systematic approach for handling physical constraints for automatic control of cyber-physical systems~\cite{Maciejowski2001,RM2009}, often leading to improved control performance and reduced tuning effort for new applications. However, the intense computational demands imposed by \ac{MPC} precludes its use in applications that could benefit considerably from its advantages, especially in those that have fast required response times and in those that must run on resource-constrained, embedded computing platforms with low cost or low power requirements. 

For linearly constrained \ac{MPC} problems of low dimensionality, one can partially avoid this computational burden by precomputing the solution map offline using multi-parametric programming~\cite{Bemporad2002}. In this case, the online controller implementation consists only of region search and table look-up procedures. Further work integrating the design of the solution map and embedded circuits has further increased the efficiency in performing these operations~\cite{Comaschi2012}. However, for larger problems, this approach quickly becomes impractical, mainly due to substantial memory requirements, forcing a return to online optimization methods. 

Recently, there has been significant interest in using first-order methods, both in the primal~\cite{Richter2012TAC} and dual domains~\cite{Richter2011CDC,Kogel2012,Giselsson2012,Annergren2012}, for the online solution of linear-quadratic MPC problems. Compared to other solution methods for \ac{QPs}  (e.g.~active-set or interior-point schemes), first-order methods do not require the solution of a linear system of equations at every iteration, which is often a limiting factor for embedded platforms with modest computational capability. This feature, coupled with the observation that medium-accuracy solutions are often sufficient for good control performance~\cite{Wang2010}, make first-order methods promising candidates for efficient, low cost \ac{MPC}. In addition, first-order methods have certain features that make them amenable to fixed-point implementation, they can be efficiently parallelized, and their simplicity invites analysis that can guide low-level implementation choices for further efficiency gains. 

There have been several recent efforts to translate innovation in optimization algorithms into practical solvers customized for \ac{MPC} problems. In terms of software, \cite{Domahidi2012,Mattingley2011} and ~\cite{Ullmann2011} describe automatic state-of-the-art code generators for interior-point and first-order solvers, respectively, whereas~\cite{Ferreau2008} describes a widely used active-set based solver. In all cases, embedded applications were the primary target, although the solvers are implemented using double precision floating-point arithmetic, which is generally not available or is very expensive in embedded computing platforms. In terms of hardware, \cite{Jerez2011,Vouzis2009,Wills2012} describe different custom computing architectures for both interior-point and active-set methods using reduced floating-point arithmetic in \ac{FPGAs}, reporting minor speed-ups or use of expensive devices to provide significant acceleration.  Although there has been some progress in accelerating the core component of these algorithms -- solvers for linear equations -- using fixed-point arithmetic~\cite{Jerez2012CDC}, extending these results to the other aspects of interior-point or active-set algorithms remains challenging.

\subsection*{Summary of contribution}

In this paper we focus on practical and theoretical issues for efficient implementation of optimization-based control systems on low cost embedded platforms. 

\begin{enumerate}
\item {\it Architectures}: We present a set of parameterized automatic generators of custom computing architectures for solving different types of MPC problems. For input-constrained problems, we describe architectures for Nesterov's fast gradient method (first described in the preliminary publication~\cite{Jerez2013}), and for state-constrained problems we consider architectures based on the alternating direction method of multipliers (ADMM). Even if these methods are conceptually very different, they share the same computational patterns and similar computing architectures can be used to implement them efficiently. These architectures are extended to support warm starting procedures and the projection operations required in the presence of soft constraints.

\item {\it Analysis}: Since for a reliable operation using fixed-point arithmetic it is crucial to prevent overflow errors, we derive theoretical results that guarantee the absence of overflow in all variables of the fast gradient method. Furthermore, we present an error analysis of both the fast gradient method and ADMM under (inexact) fixed-point computations in a unified framework. This analysis underpins the numerical stability of the methods for hardware implementations and can be used to determine {\em a priori} the minimum number of bits required to achieve a given solution accuracy specification, resulting in minimal resource usage.

\item {\it Implementation}: We derive a set of design rules for efficient implementation of the proposed methods, such as a scaling procedure for accelerating the convergence of ADMM and criteria for determining the size of the Lagrange multipliers. The proposed architectures are characterized in terms of the achievable performance as a function of the amount of resources available. As a proof of concept, generated solver instances are demonstrated for several linear-quadratic MPC problems, reporting achievable controller sampling rates in excess of 1 MHz, while the controller can be implemented on a low cost embeddable device.

\end{enumerate}

\subsection*{Outline}
The paper is organized as follows: 
After a brief summary of the general MPC formulation and the different first-order methods in Sections~\ref{mpc-setup} and \ref{solmeths}, we focus on the fixed-point analysis in Section~\ref{sec:fixanalysis}. We follow with the hardware architectures and performance evaluation in Sections~\ref{hw-sec} and~\ref{sec:benchmark}.




\section{Soft-Constrained Model Predictive Control Setup}\label{mpc-setup}

Throughout, we address control of a discrete-time \ac{LTI} system in the form
\begin{equation}\label{eqn:LTIsystem}
x^+ = Ax + Bu,
\end{equation}
where $x\in\Reals{n_x}$ is the system state and $u\in \Reals{n_u}$ is the system input.  The overall design goal is to construct a time-invariant (possibly nonlinear) static state feedback controller $\mu : \Reals{n_x}\to\Reals{n_u}$ such that $u = \mu(x)$ stabilizes the system \eqref{eqn:LTIsystem} while simultaneously satisfying a collection of state and input constraints in the time domain.

In standard design methods for constructing linear controllers for systems in the form \eqref{eqn:LTIsystem}, the bulk of the computational effort is spent \emph{offline} in identifying a suitable controller, whose \emph{online} implementation has minimal computing requirements.   The inclusion of state and input constraints renders most such design methods unsuitable.

A now standard alternative is to use \ac{MPC} \cite{Maciejowski2001,RM2009}, which moves the bulk of the required computationally effort \emph{online} and which addresses directly the system constraints.  At every sampling instant, given an estimate or measurement of the current state of the plant~$x$, an \mbox{\ac{MPC}} controller solves a constrained $N$-stage optimal control problem in the form
\ifTwoColumn
\else
	\begin{align}\label{cost function}
    	J^*(x) = &\min
        \frac{1}{2} x_N^T Q_N x_N + \frac{1}{2} 
        \sum_{k=0}^{N-1} x_{k}^T Q x_k + u_{k}^T R u_k + 2 x_k^T S u_k +
        \sum_{k=1}^N \left(\sigma_1 \cdot \mathbf{1}^T \delta_k + 
        \sigma_2 \cdot \norm{\delta_k}_2^2\right)\\
    	&\begin{aligned}
 		\text{subject to  }
		\enspace x_0 &=  x, \nonumber \\
                      x_{k+1} &= A_d x_k + B_d u_k  ,  & k&= 0,1, \ldots, N-1, \nonumber\\
                      u_k &\in \SetU, &  k&=0, 1, \ldots, N - 1, \nonumber\\
                      \left(x_k, \delta_k \right) &\in \SetXD, &  k&=1, 2, \ldots, N. 
                      \nonumber
    	\end{aligned}
	\end{align}
\fi

If an optimal input sequence $\{u_i^*(x)\}_{i=0}^{N-1}$ and state trajectory $\{x_i^*(x)\}_{i=0}^{N}$ exists for this problem given the initial state $x$, then an \ac{MPC} controller can be implemented by applying the control input $u = u_0^*(x)$.  

We will assume throughout that the system input constraint set $\SetU$ is defined as a set of interval constraints 
$\SetU \eqdef \set{u}{u_{\low} \leq u \leq u_{\upp}}$.
We assume that the system states have both free (index set \free), hard-constrained (index set \bound) and soft-constrained (index set \soft) components, i.e.\ the set $\SetXD$ in \eqref{cost function} is defined as 
\ifTwoColumn
\else
	\begin{align}\label{eq:softstate}
		\SetXD  = \set{ (x, \delta) \in \Reals{n_x} \times \Reals{\left| \soft \right|}_{+}}
		{
		 x_\free \text{ free}, \, x_{\low} \leq x_{\bound} \leq x_{\upp}, \,
		\left| x_i - x_{\cent, i} \right| \leq r_i + \delta_i, \, i \in \soft
		}, 
	\end{align}
\fi
with $x_{\cent, i} \in \Reals{}$ being the center of the interval constraint of radius $r_i > 0$ for a soft-constrained state component. The index sets $\free, \bound$ and $\soft$ are assumed to be pairwise disjoint and to satisfy $\free\cup\bound\cup\soft = \left\{ 1, 2, \ldots, n_x \right\}$. 

We assume throughout that the  penalty matrices~$(Q, Q_N) \in\mathbb{R}^{n_x\times n_x}$ are positive semidefinite, 
$R\in\mathbb{R}^{n_u \times n_u}$ is positive definite,  
 and $S\in\mathbb{R}^{n_x\times n_u}$ is chosen such that the objective function in~\eqref{cost function} is jointly convex in the states and inputs.  There is by now a considerable body of literature \cite{bib:maynerawlings2000,RM2009} describing conditions on the penalty matrices and/or horizon length $N$ sufficient to ensure that the resulting \ac{MPC} controller is stabilizing (even when no terminal state constraints are imposed), and we do not address this point further.   For stability conditions for soft-constrained problems, the reader is referred to~\cite{Zeilinger2010} and~\cite{bib:scokaert1999} and the references therein.
 
If the soft-constrained index set~$\soft$ is nonempty, then a linear-quadratic penalty on the auxiliary variables~$\delta_k \in \Reals{|S|}_+$, weighted by positive scalars $(\sigma_1,\sigma_2)$, can be added to the objective.  
In practice, soft constraints are a common measure to avoid infeasibility of the MPC problem~\eqref{cost function} in the presence of disturbances. However, there also exist hard state constraints that can always be enforced and cannot lead to infeasibility, such as state constraints arising from remodeling of input-rate constraints.  For the sake of generality we address both types of state constraints in this paper.

If $\sigma_1$ is chosen large enough, then the optimization problem \eqref{cost function} corresponds to an \emph{exact penalty} reformulation of the associated hard-constrained problem (i.e.\ one in which the optimal solution of \eqref{cost function} maintains $\delta_k = 0$ if it is possible to do so). An exact penalty formulation preserves the optimal behavior of the MPC controller when all constraints can be enforced.  We first characterize conditions under which a soft constraint penalty function for a convex optimization problem is {exact}. 
\ifTwoColumn
	\begin{theorem}[{\cite[Prop.\ 5.4.5]{bib:bertsekas_nonlinear_1999}}]\label{thm:exactpen}
\else
	\begin{theorem}[Exact Penalty Function for Convex Programming {\cite[Prop.\ 5.4.5]%
	{bib:bertsekas_nonlinear_1999}}]\label{thm:exactpen}
\fi
Consider the convex problem
\begin{align} \label{eq:orig}
f^* \eqdef &\min_{z \in \SetQ} \, f(z) \\
&\begin{aligned}
\text{subject to } & g_j(z) \leq 0 \, , \quad j = 1, 2, \ldots, r, \nonumber 
\end{aligned}
\end{align}
where $f : \Reals{n}  \rightarrow \Reals{}$ and $g_j : \Reals{n} \rightarrow \Reals{}$, $j = 1, \ldots, r$, are convex, real-valued functions and $\SetQ$ is a closed convex subset of $\Reals{n}$.  Assume that an optimal solution~$z^*$ exists with $f(z^*) = f^*$, strong duality holds and an optimal Lagrange multiplier vector $\mu^* \in \Reals{r}_{+}$ for the inequality constraints exists.
{
\renewcommand{\theenumi}{\roman{enumi}}
\begin{enumerate}
\item If $\sigma_1 \geq \| \mu^* \|_{\infty}$ and $\sigma_2 \ge 0$, then\begin{align}\label{eq:pen}
f^* = & \min_{z \in \SetQ} \, f(z) +  \sum_{j=1}^{r}\left(\sigma_1 \cdot\delta_j + \sigma_2 \cdot \delta_j^2\right) \\
&\text{subject to } g_j(z) \leq  \delta_j, \quad \delta_j \geq  0,  \quad j = 1, 2, \ldots, r. \nonumber
\end{align} \label{thm:st1} 
\item If $\sigma_1 > \| \mu^* \|_{\infty}$ and $\sigma_2 \ge 0$, the set of minimizers of the penalty reformulation in~\eqref{eq:pen} coincides with the set of minimizers of the original problem in~\eqref{eq:orig}. \label{thm:st2}
\end{enumerate}
}
\end{theorem}



\begin{remark}
In the context of the MPC problem~\eqref{cost function}, the penalty reformulation is exact if the penalty parameter~$\sigma_1$ is chosen to be greater than the largest Lagrange multiplier for any constraint $\left| x_i - x_{\cent, i} \right| \leq r_i$, $i \in \soft$, over all feasible initial states~$x$. In general, this bound is unknown \emph{a priori} and is treated as a tuning parameter in the control design.   The quadratic penalty parameter $\sigma_2$ need not be nonzero for such a penalty formulation to be exact, but the inclusion of a nonzero quadratic term is necessary for our numerical stability results under fixed-point arithmetic in Section \ref{sec:fixanalysis}. 
\end{remark}

For the sake of notational simplicity, the results of this paper are presented with reference to the optimal control problem in regulator form in~\eqref{cost function}.  However, all of our results generalize easily to setpoint tracking problems.


\section{First-Order Solution Methods}\label{solmeths}
We next describe two different first-order optimization methods for solving the optimal control problem~\eqref{cost function} efficiently.   In particular, we apply the primal fast gradient method (FGM) in cases where only input-constraints are present, and a dual method based on the alternating direction method of multipliers (ADMM) for cases in which both state- and input-constraints are present.


\subsection{Input-Constrained MPC Using the Fast Gradient Method}\label{fg-sec}
The fast gradient method is an iterative solution method for smooth convex optimization problems first published by Nesterov in the early 80s \cite{bib:nesterov1983}, which requires the objective function to be strongly convex~\cite[\S 9.1.2]{BoydBook04}. The method can be applied to the solution of MPC problem~\eqref{cost function} if the future state variables $x_i$ are eliminated by expressing them as a function of the initial state, $x$, and the future input sequence (so-called \emph{condensing} \cite{Maciejowski2001}), resulting in the problem
\begin{align}
f^*(x) = &\min_z f(z; x) \eqdef \frac{1}{2} z^TH_{F}z + z^T \Phi x \label{costfuncQP}\\
&\text{subject to }  z \in \SetK, \nonumber 
\end{align}
where $z := (u_0,\ldots,u_{N-1}) \in \mathbb{R}^n$, $n =N n_u$, the Hessian $H_{F} \in \Reals{n \times n}$ is positive definite under the assumptions in Section~\ref{mpc-setup}, and the feasible set is given as $\mathbb{K} := \mathbb{U}\times \ldots \times \mathbb{U}$. The current state only enters the gradient of the linear term of the objective through the matrix~$\Phi \in \Reals{n \times n_x}$.

We consider the \emph{constant step scheme II} of the fast gradient method in \cite[\textsection 2.2.3]{bib:nesterov2004}. Its algorithmic scheme for the solution of~\eqref{costfuncQP}, optimized for parallel execution on parallel hardware, is given in Algorithm~\ref{alg:fg2-algorithm}. Note that the state-independent terms $(I-\frac{1}{L}H_{F}) $, $\frac{1}{L} \Phi$ and $(1 + \beta)$ can all be computed offline and that the product $\frac{1}{L} \Phi x$ must only be evaluated once. The core operations in Algorithm~\ref{alg:fg2-algorithm} are the evaluation of the gradient (implicit in line~\ref{line:two}) and the projection operator of the feasible set, $\pi_{\SetK}$, in line~\ref{line:three}. Since for our application the set~$\SetK$ is the direct product of the $N$ $n_u$-dimensional sets $\SetU$, it suffices to consider $N$ independent projections that can be performed in parallel. For the specific case of a box constraint on the control input, every such projection corresponds to $n_u$ scalar projections on intervals, each computable analytically. In this case, the fast gradient method requires only multiplication and addition, which are considerably faster and use significantly less resources than division when implemented using digital circuits. 

It can be inferred from \cite[Theorem 2.2.3]{bib:nesterov2004} that for every state~$x$, Algorithm~\ref{alg:fg2-algorithm} generates a sequence of iterates $\{ z_i \}_{i = 1}^{I_{\max}}$ such that the residuals $f(z_i; {x}) - f^*(x)$ are bounded by
\begin{align}\label{eq:fgmconv}
\min \Biggl\{ \!\! \bigg( 1 - \sqrt{\frac{1}{\kappa}} \biggr)^i \!\!, 
\frac{4 \kappa}{(2 \sqrt{\kappa} + i)^2} \! \Biggr\} \! \cdot
2  \bigl( f(z_0; {x}) - f^*(x) \bigr) ,
\end{align}
for all $i = 0, \ldots, I_{\max}$, where $\kappa$ denotes the condition number of $f$, or an upper bound of it, given by $\kappa \!=\! L/\mu$, where $L$ and $\mu$ are a Lipschitz constant for the gradient of $f$ and convexity parameter of $f$, respectively. Note that the convexity parameter $f$ for a strongly convex quadratic objective function as in~\eqref{costfuncQP} corresponds to the minimum eigenvalue of $H_F$.  Based on this convergence result, which states that the  bound exhibits the best of a linear and a sublinear rate, one can derive a certifiable and practically relevant iteration bound~$I_{\max}$ such that the final residual is guaranteed to be within a specified level of suboptimality for all initial states arising from a bounded set \cite{Richter2012TAC}. It can further be proved that there is no other variant of a gradient method with better theoretical convergence \cite{bib:nesterov2004}, i.e.\ the fast gradient method is an \emph{optimal} gradient method, in theory.

\begin{algorithm}[t!]
\begin{algorithmic}[1]
\REQUIRE Initial iterate $z_0 \in \SetK$, $y_0=z_0$, upper (lower) bound~$L$ ($\mu > 0$) on maximum (minimum) eigenvalue of Hessian~$H_{F}$,
step size $\beta = \bigl( \sqrt{L} - \sqrt{\mu} \bigr) / \bigl( \sqrt{L} + \sqrt{\mu} \bigr)$
\FOR{$i=0$ to $I_{\max}-1$}

\STATE $t_{i} := (I-\frac{1}{L}H_{F}) y_i - \frac{1}{L} \Phi x$ \label{line:two}

\STATE $z_{i+1}:=\pi_{{\SetK}}(t_{i})$ \label{line:three}

\STATE $y_{i+1}:=(1+\beta) z_{i+1}-\beta z_{i}$ \label{line:four}

\ENDFOR \\
\end{algorithmic}
\caption{Fast gradient method for the solution of MPC problem~\eqref{costfuncQP} at state~$x$ (\emph{optimized for parallel hardware})}
\label{alg:fg2-algorithm}
\end{algorithm}

The fast gradient method is particularly attractive for application to MPC in embedded control system design due both to the relative ease of implementation and to the availability of strong performance certification guarantees.  However, its use is limited to cases in which the projection operation $\pi_{\SetK}$ is simple, e.g.\ in the case of box-constrained inputs.  
Unfortunately, the inclusion of {state} constraints changes the geometry of the feasible set $\SetK$ such that the projection subproblem is as difficult as the original problem, since the constraints are no longer separable in $u_k$.  In the next section we therefore describe an alternative solution method in the dual domain that avoids these complications, though at the expense of some of the strong certification advantages.

\subsection{Input- and State-Constrained MPC Using ADMM}\label{admm}

In the presence of state constraints, if one imposes~$(Q, Q_N) \in\mathbb{R}^{n_x\times n_x}$ to be positive definite, the fast gradient method can be used again to solve the dual problem via Lagrange relaxation of the  equality constraints~\cite{Richter2011CDC}. However, in this case the dual function is \emph{not} strongly concave and consequently the convergence speed is severely affected. A quadratic regularizing term can be added to the Lagrangian to improve convergence, but this prevents the use of distributed operations for computing the gradient of the dual function, adding a significant computational overhead.   We therefore seek an alternative approach in the dual domain.

For dual problems we do \emph{not} work in the condensed format~\eqref{costfuncQP}, but rather maintain the state variables $x_k$ in the vector of decision variables $z := (u_0,\dots,u_{N-1},x_0,\delta_0,\dots,x_N,\delta_{N}) \in \mathbb{R}^n$, $n=N(n_u+n_x + |\soft|)+n_x + |\soft|$, resulting in the problem
\begin{align}
f^*(x) = &\min_z  f(z; x) \eqdef \frac{1}{2} z^TH_{A}z + z^Th  \label{costfuncQPSparse}\\ 
&\text{subject to }  z \in \SetK,\,\,  Fz = b(x).\nonumber 
\end{align}

The affine constraint $Fz = b(x)$ models the dynamic coupling of the states $x_k$ and $u_k$ via the state update equation~\eqref{eqn:LTIsystem}, and is at the root of the difficulty in projecting the variables $z$ onto the constraints in the fast gradient method. 

The alternating direction method of multipliers (ADMM) \cite{Boyd2011} partitions the optimization variables into two (or more) groups to maintain the possibility of decoupled projection.   In applying ADMM to the specific problem~\eqref{costfuncQP}, we maintain an additional copy $y$ of the original decision variables $z$ and solve the problem
\ifTwoColumn
\else
	\begin{align}
		f^*(x) = \min_{z,y} f(z, y; x) &:= \frac{1}{2} y^TH_{A}y + y^T h + I_{\SetA}(y;x) +
	 	I_{\SetK}(z) + \frac{\rho}{2}\| y-z \|^2 \label{costfuncQP-ADMM}\\
		\text{subject to }  z = y,  \label{ADMMmatching}
	\end{align}
\fi
where $(z,y) \in \Reals{2n}$ contain copies of all input, state and slack variables.  The functions $I_{\SetA} : \Reals{n} \times \Reals{n_x}  \to \{0,+\infty\}$ and $I_{\SetK} : \Reals{n}  \to \{0,+\infty\}$ are indicator functions for the sets described by the equality and inequality constraints, respectively, e.g.
\begin{align}\label{equ_const}
I_{\SetA}(y,x) := \begin{cases} 0 & \text{ if } Fy=b(x) \, , \\ \infty & \text{ otherwise}  \, , \end{cases} 
\end{align}
where $\mathbb{K} \eqdef \mathbb{U}\times \ldots \times \mathbb{U} \times \SetXD \times \ldots \times \SetXD$. The current state $x$ enters the optimization problem through~(\ref{equ_const}).  The inclusion of the regularizing term $(\rho/2)\| y-z \|^2$ has no impact on the solution to~\eqref{costfuncQP-ADMM} (equivalently~\eqref{costfuncQPSparse}) due to the compatibility constraint $y=z$, but it does allow one to drop the smoothness and strong convexity conditions on the objective function, so that one can solve control problems with more general cost functions such as those with $1$- or $\infty$-norm stage costs.

Note that there are many possible techniques for copying and partitioning of variables in ADMM. In the context of optimal control, the choice given in~(\ref{costfuncQP-ADMM}) results in attractive computational structures~\cite{ODonoghue2013}.

The dual problem for~\eqref{costfuncQP-ADMM} is given by
\ifTwoColumn
\else
	\begin{align*}
	\max_{\nu} g(\nu) := \inf_{z,y} L_{\rho}(z,y,\nu) := 
	\frac{1}{2} y^TH_{A}y + y^T h + I_{\SetA}(y;x) + I_{\SetK}(z) + \nu^T(y-z) + 
	\frac{\rho}{2}\| y-z \|^2 \, .
	\end{align*}
\fi

ADMM solves this dual problem using an approximate gradient method by repeatedly carrying out the steps
\begin{subequations}\label{eqn:ADMMsketch}
\begin{align}
y_{i+1} 	& := \arg \min_y L_{\rho}(z_i,y,\nu_i)  \, , \label{unc_step} \\ 
z_{i+1} 	& := \arg \min_z L_{\rho}(z,y_{i+1},\nu_i)  \, , \label{conc_step} \\
\nu_{i+1} 	& := \nu_i + \rho (y_{i+1}-z_{i+1})   \, . \label{dual_grad} 
\end{align}
\end{subequations}
The gradient of the dual function is approximated by the expression $(y_{i+1}-z_{i+1})$ in~(\ref{dual_grad}), which employs a single Gauss-Seidel pass instead of a joint minimization to allow for decoupled computations. Choosing the regularity parameter $\rho$ also as the step-length arises from Lipschitz continuity of the (augmented) dual function.
There are at present no universally accepted rules for selecting the value of the penalty parameter however, and it is typically treated as a tuning parameter during implementation.

Our overall algorithmic scheme for ADMM for the solution of~\eqref{costfuncQP-ADMM} based on the sequence of operations~\eqref{unc_step}--\eqref{dual_grad}, optimized for parallel execution on parallel hardware, is given in Algorithm~\ref{alg:admm-algorithm}.  The core computational tasks are the equality-constrained optimization problem~\eqref{unc_step} and the inequality-constrained, but separable, optimization problem~\eqref{conc_step}.

In the case of the equality-constrained minimization step~(\ref{unc_step}), a solution can be computed from the KKT conditions by solving the linear system \begin{align*}
\begin{bmatrix}
H_{A} + \rho I & F^T \\
F             & 0
\end{bmatrix} \begin{bmatrix}
y_{i+1} \\
\lambda_{i+1}
\end{bmatrix} = \begin{bmatrix}
- h -\nu_{i} + \rho z_{i}\\
b(x)
\end{bmatrix}.
\end{align*}

Note that only the vector $y_{i+1}$, and not the multiplier $\lambda_{i+1}$, arising from the solution of this linear system is required for our ADMM method.  The most efficient method to solve for $y_{i+1}$ is to invert the (fixed) KKT matrix \emph{offline}, i.e.\ to compute
\begin{align*}
\begin{bmatrix}
M_{11} & M_{12} \\
M_{12}^T & M_{22}
\end{bmatrix} = \begin{bmatrix}
H_{A} + \rho I & F^T \\
F             & 0
\end{bmatrix}^{-1} \, ,
\end{align*}
 and then to obtain $y_{i+1}$ \emph{online} from $y_{i+1} = M_{11} \left( -h-\nu_{i} + \rho z_{i} \right) + M_{12} b(x)$ as in Line~\ref{line:two2} of Algorithm~\ref{alg:admm-algorithm}.  Observe that the product $M_{12}b(x)$ needs to be evaluated only once, and that this matrix is always invertible when $\rho > 0$ since $F$ has full row rank.

The inequality-constrained minimization step~(\ref{conc_step}) results in the projection operation in Line~\ref{line:three2} of Algorithm~\ref{alg:admm-algorithm}. In the presence of soft state constraints, this operation requires independent projections onto a truncated two-dimensional cone, which can be efficiently parallelized and require no divisions.  We describe efficient implementations of this projection operation in parallel hardware in Section~\ref{hw-sec}. 

\begin{algorithm}[t!]
\begin{algorithmic}[1]
\REQUIRE Initial iterate $z_0 = z^{*-}$, $\nu_0=\nu^{*-}$, where $z^{*-}$ and $\nu^{*-}$ are the shifted solutions at the previous time instant (see Section~\ref{hw-sec}), and $\rho$ is a constant power of 2.
\FOR{$i=0$ to $I_{\max}-1$}

\STATE $y_{i+1} := M_{11}(-h+\rho z_i-\nu_i) + M_{12}b(x)$ \label{line:two2}

\STATE $z_{i+1}:=\pi_{{\SetK}}(y_{i+1}+\frac{1}{\rho}\nu_i)$ \label{line:three2}

\STATE $\nu_{i+1}:=\rho y_{i+1} +\nu_i - \rho z_{i+1}$ \label{line:four2}

\ENDFOR \\
\end{algorithmic}
\caption{ADMM for the solution of MPC problem~\eqref{costfuncQP} at state~$x$ (\emph{optimized for parallel hardware})}
\label{alg:admm-algorithm}
\end{algorithm}

This variant of ADMM is known to converge; see~\cite[\S3.4; Prop.~4.2]{bib:bertsekas1997} for general convergence results. More recently, a bound on the convergence \emph{rate} was established in \cite{bib:admmrate2011}, where it was shown that the error in ADMM, for a different error function, decreases as $1/i$, 
where $i$ is the number of iterations.  
This result still compares unfavorably relative to the known $1/i^2$ convergence rate for the fast gradient method in the dual domain.  However, the observed convergence behavior of ADMM in practice is often significantly faster than for the fast gradient method~\cite{Boyd2011}.


\subsection{ADMM, Lagrange multipliers and soft constraints}\label{const_scaling}

Despite its generally excellent empirical performance, ADMM can be observed to converge very slowly in certain cases. In particular, for MPC problems in the form~\eqref{costfuncQP}, convergence may be very slow in those cases where there is a large mismatch between the magnitude of the optimal Lagrange multipliers $\nu^*$ for the equality constraint~\eqref{ADMMmatching} and the magnitude of the primal iterates $(z_i,y_i)$.  The reason is evident from the ADMM multiplier update step~\eqref{dual_grad}; the existence of very large optimal multipliers $\nu^*$ necessitates a large number of ADMM iterations when the difference $(z_{i}-y_{i})$ remains small at each iteration and $\rho \approx 1$. 

This effect is of particular concern for MPC problem instances with soft constraints.  If one denotes by $z_\delta$ those components of $z$ associated with the slack variables  
$\{\delta_1,\dots,\delta_N\}$ (with similar notation for $y_\delta$), then the objective function~\eqref{costfuncQP-ADMM} features a term $\sigma_1\cdot\mathbf{1}^Ty_\delta$, with the exact penalty term $\sigma_1$ typically very large.  The equality constraints~\eqref{ADMMmatching} include the matching condition $z_\delta - y_\delta = 0$, with associated Lagrange multiplier $\nu_\delta$.  Recalling the usual sensitivity interpretation of the optimal multiplier $\nu_\delta^*$, one can conclude that  $\nu_\delta^* \approx \sigma_1\cdot\mathbf{1}$ in the absence of unusual problem scaling\footnote{If one sets the regularization parameter $\rho = 0$ in~\eqref{costfuncQP-ADMM} and $\sigma_2 = 0$, then it can be shown that this approximation becomes exact.}.   

For soft constrained problems, we avoid this difficulty by rescaling those components of the matching condition~\eqref{ADMMmatching} to the equivalent condition $(1/\sigma_1)(z_\delta - y_\delta) = 0$, which results in a rescaling of the associated optimal multipliers to $\nu_\delta^* \approx 1$.  The aforementioned convergence difficulties due to excessively large optimal multipliers are then avoided.

\section{Fixed-Point Aspects of First-Order Solution Methods}\label{sec:fixanalysis}

In this section we  first motivate the use of fixed-point arithmetic from a hardware efficiency perspective and then isolate potential error sources under this arithmetic. We concentrate on two types of errors. For \emph{overflow errors} we provide analysis to guarantee that they cannot occur in the fast gradient method, whereas for \emph{arithmetic round-off errors} we prove that there is a converging upper bound on the total incurred  error in either of the two methods. The results we obtain hold under the assumptions in Section~\ref{sssec:assume_fgm} and guarantee reliable operation of  first-order methods on fixed-point platforms.

\subsection{Fixed-Point Arithmetic and Error Sources}\label{ssec:fixvsfloat}
Modern computing platforms must allow for a wide range of applications that operate on data with potentially large dynamic range, i.e.\ the ratio of the smallest to largest number to be represented. For general purpose computing, \emph{floating-point} arithmetic provides the necessary flexibility. A floating-point number consists of a sign bit, a mantissa, and an exponent value that moves the binary point with respect to the mantissa. The dynamic range grows doubly exponentially with the number of exponent bits, making it possible to represent a wide range of numbers with a relatively small number of bits. However, because two operands can have different exponents, it is necessary to perform \emph{denormalization} and \emph{normalization} operations before and after every addition or subtraction, leading to increased resource usage and long arithmetic delays. 

In contrast, hardware platforms employing \emph{fixed-point} numbers use a fixed number of bits for the integer and fraction fields, i.e.\ the exponent does not vary and does not have to be stored. Fixed-point computations are the same as with integer arithmetic, hence the digital circuitry is simple and fast, leading to greater power efficiency and significant potential for acceleration via extra parallelization in a custom hardware implementation. For instance, in a typical modern FPGA platform~\cite{logicoreFP_v50} fixed-point addition takes one clock cycle, 
whereas a single precision floating-point adder would require 14 cycles 
while using one order of magnitude more resources for the same number of bits. 

The benefits of fixed-point arithmetic motivate its use in first-order methods to realize fast and efficient implementations of Algorithms~\ref{alg:fg2-algorithm}  and~\ref{alg:admm-algorithm} on FPGAs or other low cost and low power devices with no floating-point support, such as embedded microcontrollers, fixed-point \ac{DSPs} or \ac{PLCs}. However, reduced precision representations and fixed-point computations incur several types of errors that must be accounted for.  These include:

\subsubsection*{Quantization Errors} Finite representation errors arise when converting the problem and algorithm data from high precision to reduced precision data formats. Potential consequences include  loss of problem convexity, change of optimal solution and a lack of feasibility with respect to the original problem.
\subsubsection*{Overflow Errors} Overflow errors occur whenever the number of bits for the integer part in the fixed-point representation is too small, and can cause unpredictable behavior of the algorithm.
\subsubsection*{Arithmetic Errors} Unlike with floating-point arithmetic, fixed-point addition and subtraction operations involve no round-off error provided there is no overflow and the result has the same number of fraction bits as the operands~\cite{Wilkinson63}. For multiplication, the exact product of two numbers with $b$ fraction bits can be represented using $2b$ fraction bits, hence a $b$-bit truncation of a 2's complement number incurs a round-off error bounded from below by $-2^{-b}$. Recall that in 2's complement arithmetic, truncation incurs a negative error both for positive and negative numbers. 


\subsection{Notation and Assumptions}\label{sssec:assume_fgm}

We will use $\hat{(\cdot)}$ throughout in order to distinguish quantities in a fixed-point representation from those in an exact representation and under exact arithmetic.   Throughout, we assume for simplicity that all variables and problem data are represented using the same number of fraction bits $b$.  We further assume that the feasible sets under finite precision satisfy $\widehat{\SetK} \subseteq \SetK$, so that solutions in fixed point arithmetic do not produce infeasibility in the original problem due to quantization error.

We conduct separate analyses of both overflow and arithmetic errors for the fast gradient method (Algorithm \ref{alg:fg2-algorithm}) and ADMM (Algorithm \ref{alg:admm-algorithm}). In both cases, the central requirement is to choose the number of fraction bits $b$ large enough to ensure satisfactory numerical behavior.  We therefore employ two different sets of assumptions depending on the numerical method in question:

\begin{assumption}[Fast Gradient Method / Algorithm \ref{alg:fg2-algorithm}]\label{assum:FGM}
The number of fractions bits $b$ and a constant $c\ge 1$ are chosen large enough such that
\begin{enumerate}[i)]

\item \label{assum:FGM:conv} 
The matrix 
\begin{align*}
H_{n} = \frac{1}{c \cdot \lambda_{\max}(\hat{H}_F)} \cdot \hat{H}_F,
\end{align*}
has a fixed-point representation $\hat{H_{n}}$ with all of its eigenvalues in the interval $(0,1]$, where $\hat{H}_F$ is the fixed-point representation of the Hessian~$H_F$, with $\lambda_{\max}(\hat{H}_F)$ its maximum eigenvalue.

\item\label{assum:FGM:beta}
The fixed-point step size~$\hat{\beta}$ satisfies
\begin{align*}
1 > \hat{\beta} \geq \Bigl({\sqrt{\kappa\bigl(\hat{H}_{n}\bigr)} - 1}\Bigr) \Bigl({\sqrt{\kappa\bigl(\hat{H}_{n}\bigr)} + 1}\Bigr)^{-1} \geq 0 \, ,
\end{align*}
where $\kappa(\hat{H}_{n})$ is the condition number of $\hat{H}_{n}$.

\end{enumerate}
\end{assumption}

\begin{assumption}[ADMM / Algorithm \ref{alg:admm-algorithm}]\label{assume:ADMM}
The number of fractions bits $b$ is chosen large enough such that
\begin{enumerate}[i)]\label{assum:ADMM}

\item \label{assum:ADMM:M11_M12}
The matrix
\begin{align*}
\left(
\begin{bmatrix}
\hat{M}_{11} & \hat{M}_{12} \\
\hat{M}_{12}^T & M_{22}
\end{bmatrix}^{-1} - 
\begin{bmatrix}
\rho I & \hat F^T \\
\hat F & 0
\end{bmatrix}
\right)
\end{align*}
is positive semidefinite, where $\rho$ is chosen such that it is exactly representable in $b$ bits.

\item \label{assum:ADMM:model_feas} The quantization errors in the matrix $\hat{F}$ are insignificant compared to the model uncertainty.

\end{enumerate}
\end{assumption}

Observe that it is always possible to select $b$ sufficiently large to satisfy all of the preceding assumptions, implying that the above conditions represent a lower bound on the number of fraction bits required in a fixed-point implementation of our two algorithms to ensure that our stability results are valid.  
Assumptions \ref{assum:FGM}.\eqref{assum:FGM:conv} and \ref{assum:ADMM}.\eqref{assum:ADMM:M11_M12} ensure that the objective functions \eqref{costfuncQP} (for the fast gradient method) and \eqref{costfuncQP-ADMM} (for ADMM) remain strongly convex and convex, respectively, despite any quantization error.

In the case of the fast gradient method, Assumption~\ref{assum:FGM}.\eqref{assum:FGM:beta} guarantees that the \emph{true} condition number of $\hat{H}_n$ is not underestimated, in which case the convergence result of the fast gradient method in~\eqref{eq:fgmconv} would be invalid. In fact, the assumption ensures that the effective condition number for the convergence result is given by
\begin{align}\label{eq:effcond}
\kappa_{n} = \biggl( \frac{1 + \hat{\beta}}{1 - \hat{\beta}} \biggr)^2 \geq \kappa\bigl(\hat{H}_n\bigr).
\end{align}
%
%
\subsection{Overflow Errors}\label{sssec:overflow_fgm}

In order to avoid overflow errors in a fixed-point implementation, the largest absolute values of the iterates' and intermediate variables' components must be known or upper-bounded \emph{a priori} in order to determine the number of bits required for their \emph{integer parts}.
For the \emph{static} problem data $(I - \hat{H}_n)$, $\hat{\Phi}_n$, $1+ \hat{\beta}$, $\hat{\beta}$, $\hat{M}_{11}$, or $\hat{M}_{12}$, the number of integer bits is easily determined 
by the maximum absolute value in each expression. 

\subsubsection{Overflow Error Bounds in the Fast Gradient Method}~

In the case of the fast gradient method, it is possible to bound analytically the largest absolute values of all of the dynamic data, i.e.\ the variables that change with every iteration.   We will 
denote by $\hat{\Phi}_{n}$ the fixed-point representation of
\begin{align*}
{\Phi}_{n} = \frac{1}{c \cdot \lambda_{\max}(\hat{H}_F)} \cdot \Phi .
\end{align*}

We summarize the upper bounds on variables appearing in the fast gradient method in the following proposition.%
\begin{proposition}\label{prop:bounds}
If problem~\eqref{costfuncQP} is solved by the fast gradient method using the appropriately adapted Algorithm~\ref{alg:fg2-algorithm}, then the largest absolute values of the iterates and intermediate variables are given by
\begin{align}
\| \hat{z}_{i+1} \|_{\infty} &\leq \bar{z} := \max \left\{ \| \hat{z}_{\min} \|_{\infty}, \| \hat{z}_{\max} \|_{\infty} \right\}, \nonumber \\
\| \hat{y}_{i+1} \|_{\infty} &\leq \bar{y} := \bar{z} + \hat{\beta} \| \hat{z}_{\max} - \hat{z}_{\min} \|_{\infty}, \nonumber \\
\| (I - \hat{H}_n) \, \hat{y}_i \|_{\infty} &\leq \bar{y}_{\text{inter}} := \| I - \hat{H}_n \|_{\infty} \cdot \bar{y}, \label{eq:matrixvec}\\ 
\| \hat{x} \|_{\infty} &\leq \bar{x} := \max_{x \in \widehat{\SetX}_0} \| x \|_{\infty}, \nonumber\\
\| \hat{\Phi}_n \hat{x} \|_{\infty} &\leq \bar{h} := \| \hat{\Phi}_n \|_{\infty} \cdot \bar{x}, \text{ and} \nonumber\\
\| t_i \|_{\infty} &\leq \bar{t} := \bar{y}_{\text{inter}} + \bar{h}, \nonumber
\end{align}
for all  $i = 0, 1, \ldots, I_{\max}-1$. The set~$\widehat{\SetX}_0$ is chosen such that for every state in exact arithmetic $x \in \SetX_0$ we have $\hat{x} \in \widehat{\SetX}_0$.
\end{proposition}
\begin{proof}
Follows from interval arithmetic and properties of the vector/matrix $\| \cdot \|_{\infty}$-norm.
\end{proof}

Note that normalization of the objective as introduced in Section~\ref{sssec:assume_fgm} has no effect on the maximum absolute values of the iterates. Furthermore, the bound in~\eqref{eq:matrixvec} also applies for the intermediate elements/cumulative sums in the evaluation of the matrix-vector product. Observe that most of the bounds stated in Proposition~\ref{prop:bounds} are tight.

\subsubsection{Overflow Error Bounds in ADMM}~

If problem~(\ref{costfuncQP-ADMM}) is solved using ADMM via Algorithm~\ref{alg:admm-algorithm}, then we do not know of any general method to upper bound the Lagrange multiplier iterates $\nu_i$ analytically, and consequently are unable to establish analytic upper bounds on all expressions involving {\em dynamic} data. In this case, one must instead estimate the undetermined upper bounds through simulation and add a safety factor when allocating the number of integer bits. As a result, with ADMM, we trade analytical guarantees on numerical behavior for the capability to solve more general problems.

\subsection{Arithmetic Round-Off Errors}\label{sssec:arithmeticError}

We next derive an upper bound on the deviation of an optimal solution $\hat z^*$ produced via a fixed-point implementation of either Algorithm~\ref{alg:fg2-algorithm} or~\ref{alg:admm-algorithm} from the optimal solutions produced from the same algorithms implemented using exact arithmetic.  In both cases, we denote by $\hat{z}_i$ a \emph{fixed-point} iterate. We wish to relate these iterates to the iterates $z_i$ generated under \emph{exact arithmetic}, by establishing a bound in the form
\[
\norm{\hat{z_i}-z_i} = \norm{\eta_i} \le \Delta_i
\]
with $\lim_{i\to\infty} \Delta_i$ finite, where $\eta_i := \hat{z_i}-z_i$ is the solution error attributable to arithmetic round-off error up to the $i^{th}$ iteration.  Consequently, we can show that inaccuracy in the computed optimal solution induced by arithmetic errors in either algorithm are bounded, which is a crucial prerequisite for reliable operation of first-order methods on fixed-point platforms.  

In both cases, we use a control-theoretic approach based on standard Lyapunov methods to derive bounds on the solution error arising specifically from fixed-point arithmetic error.  For simplicity of exposition, we consider only those errors arising from arithmetic errors and neglect quantization errors in the analysis.  This choice does not alter substantively the results presented for either algorithm.    %
Our approach is in contrast to (and more direct than) other approaches to error accumulation analysis in the fast gradient method such as~\cite{bib:baes2009,bib:schmidt2011}, which consider inexact gradient computations but do not address arithmetic round-off errors explicitly.  In the case of ADMM, we are not aware of any existing results relating to error accumulation in fixed-point arithmetic.

\subsubsection{Stability of Arithmetic Errors in the Fast Gradient Method}~

We consider first the numerical stability of the fast gradient method, by examining in detail the arithmetic error introduced at each step of a fixed-point implementation of Algorithm~\ref{alg:fg2-algorithm}.  

At iteration~$i$, the error in line~\ref{line:two} of Algorithm~\ref{alg:fg2-algorithm} is given by
\begin{align*}
\hat{t}_i - t_i = (I - \hat{H}_n) (\hat{y}_i - y_i) + \epsilon_{t,i}\, ,
\end{align*}
where $\epsilon_{t,i}$ is a vector of errors from the matrix-vector multiplication. Since there are $n$ round-off errors in the computation of every component, $\epsilon_{t,i}$ is componentwise in the interval $[ -n 2^{-b}, 0 ]$.

For the projection in line~\ref{line:three}, and recalling that $\widehat{\SetK} \subseteq \SetK$ is a box, no arithmetic error is introduced.  Indeed, one can easily verify that the error $\hat{t}_i - t_i$ can only be reduced by 
multiplication with a diagonal matrix~$\diag{\epsilon_{\pi,i}}$, with $\epsilon_{\pi,i}$ componentwise in the interval $[0,1]$. 

Finally, in line~\ref{line:four}, the error induced by fixed-point arithmetic is
\begin{align*}
\hat{y}_{i+1} - y_{i+1} = (1 + \hat{\beta}) \eta_{i+1} - \hat{\beta} \eta_{i} + \epsilon_{y, i}\, , 
\end{align*}
where two scalar-vector multiplications incur error $\epsilon_{y, i}$ with components in $[ -2^{-b}, 2^{-b} ]$ (addition \emph{and} subtraction). Defining the initial error residual terms $\eta_{-1} = \eta_0 = \hat{z}_0 - z_0$, and setting  $\hat{z}_0 - z_0 = \hat{y}_0 - y_0$, one can derive the two-step recurrence
\begin{align*}
\eta_{i+1} = \diag{\epsilon_{\pi,i}} \!  \bigl( I \!-\! \hat{H}_n \bigr) \bigl( \eta_i \!+\! \hat{\beta} (\eta_i \!-\! \eta_{i-1}) \!+\! \epsilon_{y, i-1} \bigr) \!+ \epsilon_{t, i} 
\end{align*}
for the accumulated arithmetic error at each iteration.    Note that the error~$\eta_{i}$ at each iteration is inherently bounded by the box~$\widehat{\SetK}$.  However, in the absence of the projection operation of line~\ref{line:three} and the associated error truncation, these errors remain bounded. To show this, we can express the evolution of the arithmetic error using the two-step recurrence
\begin{align}\label{eq:onestep_fg}
\underbrace{\begin{bmatrix}
\eta_{i+1} \\
\eta_{i}
\end{bmatrix}}_{=: \xi_{i+1}} =
&  \underbrace{\begin{bmatrix}
\bigl( 1 + \hat{\beta} \bigr) \bigl( I - \hat{H}_n \bigr)  &   - \hat{\beta} \bigl( I - \hat{H}_n \bigr)  \\
I   &   0
\end{bmatrix}}_{=: A}
\underbrace{\begin{bmatrix}
\eta_{i} \\
\eta_{i-1}
\end{bmatrix}}_{\xi_{i}} \nonumber \\
&+
\underbrace{
\begin{bmatrix}
 \bigl( I - \hat{H}_n \bigr)   &  I \\
0   &   0
\end{bmatrix}}_{=: B}
\underbrace{\begin{bmatrix}
\epsilon_{y, i-1} \\
\epsilon_{t, i}
\end{bmatrix}}_{=: \upsilon_i},
\end{align}
and then show that this linear system is stable.  Recalling Assumption~\ref{assum:FGM}, which bounds the eigenvalues of $\hat{H}_n$ in the interval $(0,1]$ and $\hat\beta$ in the interval $[0,1)$, we can use the following result:

\begin{lemma}\label{lem:FGMstable}
Let~$\Hm$ be any symmetric positive definite matrix with maximum eigenvalue less than or equal to one. For every constant~$\bm$ in the interval $[0,1]$ the matrix 
\begin{align*}
\Am = \begin{bmatrix}
( 1 + \bm ) ( I - \Hm )    &   - \bm ( I - \Hm )  \\
I   &   0
\end{bmatrix}
\end{align*}
is Schur stable, i.e.~its spectral radius~$\rho(\Am)$ is less than one.
\end{lemma}
\begin{proof}
Assume the eigenvalue decomposition $I-C = V^T\Lambda V$, with $\Lambda$ diagonal with entries $\lambda_i \in [0,1)$.  The eigenvalues of $\Am$ are unchanged by left- and right-multiplication by $\left[\begin{smallmatrix}V & \\& V\end{smallmatrix}\right]$ and its transpose.  It is therefore sufficient to examine instead the spectral radius of
\[
\Am_D = \begin{bmatrix}
( 1 + \bm ) \Lambda     &   - \bm \Lambda  \\
I   &   0
\end{bmatrix}.
\]
Since this matrix has exclusively diagonal blocks, its eigenvalues coincide with those of the two-by-two submatrices
\[
\Am_{D,i} = \begin{bmatrix}
( 1 + \bm )\lambda_i     &   - \bm\lambda_i  \\
1   &   0
\end{bmatrix}, \quad \text{for } i = 1, \ldots, n,
\]
and it is sufficient to prove that every such submatrix has spectral radius less than one. Note that the eigenvalues of $\Am_{D,i}$ are the roots of the characteristic equation 
\begin{equation}\label{eqn:quadraticRoots}
\mu^2 - (1+\bm)\lambda_i\mu + \lambda_i\bm = 0.
\end{equation}
It is easily verified that a sufficient condition for any quadratic equation in the form
\[
x^2 + 2bx + c = 0
\]
to have roots strictly inside the unit disk 
is for its coefficients to satisfy i) $|b| < 1$, ii) $c < 1$ and iii) $2|b| < c+1$.  For the eigenvalue solutions to \eqref{eqn:quadraticRoots}, this amounts to i) $(1+\bm)\lambda_i/2 \!<\! 1$, ii) $\lambda_i\gamma < 1$ and iii) $(1+\bm)\lambda_i < \gamma\lambda_i+1$.  All three conditions are easily confirmed for the case $\lambda_i \in [0,1)$, $\bm \in [0,1]$.
\end{proof}

\subsubsection{Stability of Arithmetic Errors in ADMM}~

As in the preceding section, for ADMM one can analyze in detail the arithmetic error introduced at each step of a fixed-point implementation of Algorithm~\ref{alg:admm-algorithm}.  

Defining $\eta_i := \hat{z}_i - z_i$, $\gamma_i := \hat{\nu}_i - \nu_i$, a similar analysis to that of the preceding section produces the two-step error recurrence
\ifTwoColumn
\else
	\begin{align}\label{eq:onestep_admm}
	\underbrace{\begin{bmatrix}
	\eta_{i+1} \\
	\gamma_{i+1}
	\end{bmatrix}}_{=: \xi_{i+1}} =
	&  \underbrace{\begin{bmatrix}
	\rho \diag{\epsilon_{\pi,i}} \hat{M}_{11}  &   
	\!-\!\diag{\epsilon_{\pi,i}} (\hat{M}_{11}-\frac{1}{\rho}I)   \\
	\rho^2 \hat{M}_{11}(I-\diag{\epsilon_{\pi,i}})   &   
	(I- \rho \hat{M}_{11})(I-\diag{\epsilon_{\pi,i}})
	\end{bmatrix}}_{=: A}
	\underbrace{\begin{bmatrix}
	\eta_{i} \\
	\gamma_{i}
	\end{bmatrix}}_{\xi_{i}} \nonumber \\
	&+
	\underbrace{
	\begin{bmatrix}
	 \diag{\epsilon_{\pi,i}}  & 0\\
	\rho (I-\diag{\epsilon_{\pi,i}}) & I  
	\end{bmatrix}}_{=: B}
	\underbrace{\begin{bmatrix}
	\epsilon_{y, i} \\
	\epsilon_{\nu, i}
	\end{bmatrix}}_{=: \upsilon_i},
	\end{align}
\fi
where: $\epsilon_{y,i} \in [ -n 2^{-b}, 0 ]^n$ is a vector of multiplication errors arising from Algorithm \ref{alg:admm-algorithm}, line~\ref{line:two}; 
$\epsilon_{\pi, i} \in [  0,1 ]^n$ is a vector of error reduction scalings arising from the projection operation in line~\ref{line:three}; and  $\epsilon_{\nu,i} \!\in\! [ -2^{-b}, 2^{-b} ]^n$ a vector multiplication errors arising from \ref{line:four} with $\epsilon_{\nu,-1} = 0$. Note that one can show that even when $\widehat{\SetK}$ is not a box in the presence of soft state constraints, the error can only be reduced by the projection operation. The initial iterates of the recurrence relation are $\eta_{-1} = \eta_0$, where $\eta_0 := \hat{z}_0 - z_0$. 

As in the case of the fast gradient method, these arithmetic errors are inherently bounded by the constraint set $\widehat{\SetK}$.  In the absence of these bounding constraints (so that $\diag{\epsilon_{\pi,i}} = I$), one can still establish that the arithmetic errors are bounded via examination of the eigenvalues of the matrix
\begin{align}\label{ADMM_err_mat}
N := \begin{bmatrix}
\rho \hat{M}_{11}  &   -  (\hat{M}_{11}-\frac{1}{\rho}I)   \\
0   &  0
\end{bmatrix}.
\end{align}
Recalling Assumption~\ref{assum:ADMM}, we have the following result:
\begin{lemma}\label{lem:ADMMstable}
The matrix $N$ in \eqref{ADMM_err_mat} is Schur stable for any $\rho > 0$.  
\end{lemma}
\begin{proof}
The eigenvalues of~(\ref{ADMM_err_mat}) are either $0$ or $\rho\lambda_i(\hat{M}_{11})$, so it is sufficient to show that the symmetric matrix $\hat{M}_{11}$ satisfies $\rho\|\hat M_{11}\| < 1$.    Recalling that 
\[
\begin{bmatrix}
\hat{M}_{11} & \hat{M}_{12} \\
\hat{M}_{12}^T & \hat{M}_{22}
\end{bmatrix} = \begin{bmatrix}
\hat{Z} & \hat{F}^T \\
\hat{F}             & 0
\end{bmatrix}^{-1}
\]
where $\hat{Z} := \hat{H}_{A} + \rho I \succ 0$,  the matrix inversion lemma provides the identity
\begin{align}
\hat{M}_{11} &= \hat{Z}^{-\half}\left[I - \hat{Z}^{-\half}\hat{F}^T(\hat{F}\hat{Z}^{-1}\hat{F}^T)^{-1}\hat{F}\hat{Z}^{-\half}\right]\hat{Z}^{-\half} \nonumber\\
			 &=: \hat{Z}^{-\half}\hat{P}\hat{Z}^{-\half}, \label{M:MIL}
\end{align}
where $\hat{P}$ is a projection onto the kernel of $\hat{F}\hat{Z}^{-\half}$, hence $\|\hat M_{11}\| \le \|\hat{Z}^{-\half}\|\|{\hat{P}}\|\|\hat{Z}^{-\half}\| = \|{\hat{Z}^{-1}}\|$.  It follows that
\[
\rho\|\hat{M}_{11}\| \le \rho\|(\hat{H}_A + \rho I)^{-1}\| \le \rho\cdot\frac{1}{\lambda_{\min}(\hat{H}_A) + \rho} \le 1,
\]
where $\lambda_{\min}(\hat{H}_A)$ is the smallest eigenvalue of the positive semidefinite matrix $\hat{H}_A$.  If $\hat{H}_A$ is actually positive definite, then the preceding inequality is strict and the proof is complete.

Otherwise, to prove that the inequality is strict we must show that $1/\rho$ is not an eigenvalue for $\hat{M}_{11}$ (which is positive semidefinite by virtue of \eqref{M:MIL}).  Assume the contrary, so that there exists some eigenvector $v$ of $\hat{M}_{11}$ with eigenvalue $1/\rho$, and some additional (arbitrary) vector $q$ that solves the linear system
\[
\begin{bmatrix}
v\\q
\end{bmatrix} = 
\begin{bmatrix}
\hat{Z} & \hat{F}^T \\
\hat{F}             & 0
\end{bmatrix}^{-1}
\begin{bmatrix}
\rho \cdot v\\0
\end{bmatrix}.
\]

Any solution must then satisfy both $\hat{H}_Av \in \image(\hat{F}^T)$ and $v\in \kernel(\hat{F})$.  Consequently $v^T\hat{H}_Av = 0$, which requires $v\in\kernel(\hat{H}_A)$ since $\hat{H}_A$ is positive semidefinite.  Recall that any such $v$ can be decomposed into $v = (u_0,\dots,u_{N-1},x_0,\delta_0,\dots,x_N,\delta_{N})$.  If the quadratic penalty for each $\delta_i$ is positive definite, then  $v\in\kernel(\hat{H}_A)$ requires each $\delta_i = 0$.

Since $\hat{F}v = 0$, the remaining components of $v$ must correspond to a sequence of state and inputs compatible with the system dynamics in \eqref{cost function}, starting from an initial state $x_0 = 0$. Any solution $v\neq 0$ would then require at least one component $u_i\neq 0$. Then $v^T \hat{H}_Av \ge u_i^T Ru_i > 0$ since $R$ is assumed positive definite, a contradiction.

\end{proof}

\subsubsection{Arithmetic Errors Bounds for the Fast Gradient Method and ADMM}~

Finally, for both the fast gradient method and ADMM we can use Lemmas
\ref{lem:FGMstable} and \ref{lem:ADMMstable} to establish an upper bound on the magnitude of error~$\eta_i$ for \emph{any} arithmetic round-off errors that might have occurred up to iteration~$i$. %
\begin{proposition}\label{lem:etabound}
Let $b$ be the number of fraction bits and $n$ be the dimension of the decision vector. Consider the error dynamics due to arithmetic round-off in~\eqref{eq:onestep_fg} or in~(\ref{eq:onestep_admm}), assuming no error reduction from projection.
The magnitude of \emph{any} {accumulation of round-off errors} up to iteration~$i$, $\| \eta_i \| = \| \hat{z}_i - z_i \|$,  is upper-bounded by
\begin{align}\label{eq:etabound}
\bar{\eta}_i \!=\! \| E A^i \| \biggl\| \! \begin{bmatrix} \eta_0 \\ \eta_0 \end{bmatrix} \! \biggr\| \!+\! 2^{-b} \! \sqrt{n (1\!+\!n^2)}  \sum_{k=0}^{i-1}\! \|E A^{i-1-k} B\|\!
\end{align}
for all $i = 0, \ldots, I_{\max}-1$, where matrix $E = \begin{bmatrix} I & 0 \end{bmatrix}$.
\end{proposition}
\begin{proof}
From the one-step recurrence~\eqref{eq:onestep_fg} or~(\ref{eq:onestep_admm}) we find that
\begin{align*}
\xi_{i} = A^{i} \, \xi_0  +  \sum_{k=0}^{i-1} A^{i-1-k} B \upsilon_k, \quad  i = 0, 1, \ldots I_{\max}-1,
\end{align*}
such that the result is obtained from applying the properties of the matrix norm. Observe that $2^{-b} \sqrt{n (1 + n^2)}$ is the maximum magnitude of $\upsilon_k$ for any $k = 0, \ldots, i-1$.
\end{proof}

Since the matrix~$A$ is Schur stable, the bound in~\eqref{eq:etabound} converges.  Indeed, the effect of the initial error~$\xi_0$ decays according to
\begin{align}\label{eq:conv1}
\| E A^i \| \propto \rho(A)^i,
\end{align}
whereas the term driven by arithmetic round-off errors in every iteration behaves according to
\begin{align}\label{eq:conv2}
\sum_{k=0}^{i-1} \|E A^{i-1-k} B\| \propto \frac{1}{1 - \rho(A)} - \frac{\rho(A)^i}{1 - \rho(A)}.
\end{align}
This result can be used to choose the number of bits $b$ {\em a priori} to meet accuracy specifications on the minimizer.

\section{Embedded Hardware Architectures for First-Order Solution Methods}\label{hw-sec}

Amdahl's law~\cite{Amdahl67} states that the potential acceleration of an optimization algorithm through parallelization is limited by the fraction of sequential dependencies in the algorithm. First-order optimization methods such as the fast gradient method and ADMM have a smaller number of sequential dependencies than interior-point or active-set methods. In fact, a very large fraction of the computation involves a single readily parallelizable matrix-vector multiplication, hence the expected benefit from parallelization is substantial.  Our implementations of both the fast gradient method (Algorithm~\ref{alg:fg2-algorithm}) and ADMM (Algorithm~\ref{alg:admm-algorithm}) differ somewhat from more conventional implementations of these methods in order to minimize sequential dependencies. Observe that in both of our algorithms, the computations of the individual vector components are independent and the only communication occurs during matrix-vector multiplication. This allows for efficient parallelization given the custom computing and communication architectures discussed next. 
Specifically, we describe a tool that takes as inputs the data type, number of bits, level of parallelism and the delays of an adder/subtracter ($l_A$) and multiplier ($l_M$) and automatically generates a digital architecture described in the VHDL hardware description language.

\subsection{Hardware Architecture for the Fast Gradient Method}

For a fixed-point data type, the parameterized architecture implementing Algorithm~\ref{alg:fg2-algorithm} for problem~\eqref{costfuncQP} is depicted in Figure~\ref{fg-cirucit}. The matrix-vector multiplication is computed in the block labeled 
$\hat{v}^T\hat{w}$, which is shown in detail in Figure~\ref{fig:dot_product}. It consists of an array of $N n_u$ parallel multipliers followed by an adder reduction tree of depth $\lceil \log_2 N n_u \rceil$.
The architecture for performing the projection operation on the set~$\widehat{\SetK}$ is shown in Figure~\ref{fig:projection}. It compares the incoming value with the upper and lower bounds for that component. Based on the result, the component is either saturated or left unchanged. 

The amount of parallelism in the circuit is parameterized by the parameter~$P$. In Figure~\ref{fg-cirucit}, $P\!=\!1$, meaning that there is parallelism within each dot-product but the that $N n_u$ dot-products required for matrix-vector multiplication are computed sequentially. If the level of parallelization is increased to $P\!=\!2$, there will be two copies of the shaded circuit in Figure~\ref{fg-cirucit} operating in parallel, one computing the odd components of $\hat{y}_i$ and $\hat{z}_i$, the other computing the even. The different blocks communicate through a serial-to-parallel shift register that accepts $P$ serial streams and outputs~$N n_u$ parallel values for matrix-vector multiplication. These $N n_u$ values are the same for all blocks. It takes $\left\lceil \frac{N n_u}{P} \right\rceil$ clock cycles to have enough data to start a new iteration, hence the number of clock cycles needed to compute one iteration of the fast gradient method for $P \in \{1, \ldots, N n_u \}$ is
\begin{equation}\label{latency expression}
L_F:=\left\lceil \frac{N n_u}{P} \right\rceil + l_{A}\lceil \log_2 N n_u \rceil + 2l_{M} +3l_{A} + 1 \, .
\end{equation}

\ifTwoColumn
\else
	\ifTechReport
		\begin{figure}
		\centering
		\setlength{\unitlength}{1cm}
		{\normalsize
		\begin{picture}(14,5)(0,0)
    			\put(0,0){\includegraphics[width=0.8\columnwidth]{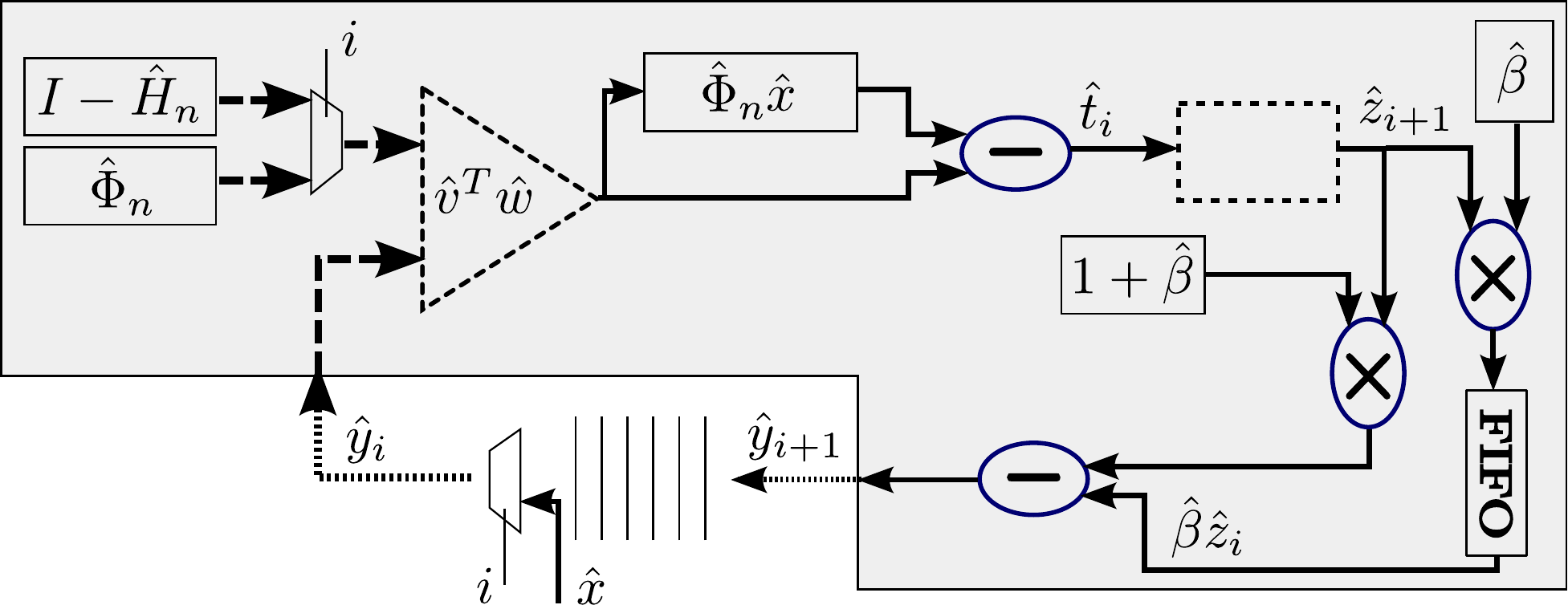}}
    			\put(10,3.6){\Large$\pi_{\widehat{\SetK}}$}
		\end{picture}
		}
		\caption{Fast gradient compute architecture. Boxes denote storage elements and dotted lines 	represent $N n_u$ parallel vector links. The dot-product block $\hat{v}^T \hat{w}$ and the projection block $\pi_{\widehat{\SetK}}$ are depicted in Figures~\ref{fig:dot_product} and~\ref{fig:projection} in detail. FIFO stands for first-in first-out memory and is used to hold the values of the current iterate for use in the next iteration. In the initial iteration, the multiplexers allow $\hat{x}$ and $\hat{\Phi}_n$ through and the result~$\hat{\Phi}_n\hat{x}$ is stored in memory. In the subsequent iterations, the multiplexers allow $\hat{y}_i$ and $I-\hat{H}_n$ through and $\hat{\Phi}_n\hat{x}$ is read from memory.}
		\label{fg-cirucit}
		\end{figure}	
	\else
		\begin{figure}
		\centering
		\setlength{\unitlength}{1cm}
		{\normalsize
		\begin{picture}(14,5)(0,0)
    			\put(0,0){\includegraphics[width=0.8\columnwidth]{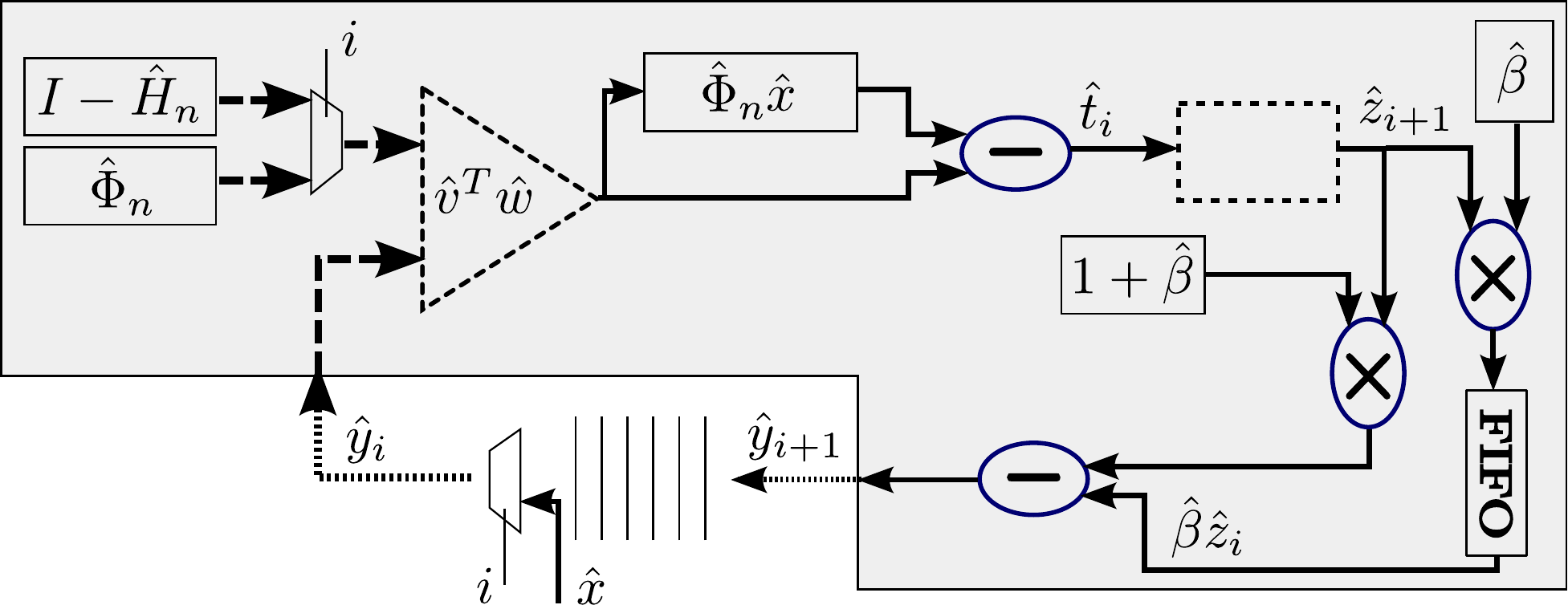}}
    			\put(10.35,3.7){\Large$\pi_{\widehat{\SetK}}$}
		\end{picture}
		}
		\caption{Fast gradient compute architecture. Boxes denote storage elements and dotted lines 	represent $N n_u$ parallel vector links. The dot-product block $\hat{v}^T \hat{w}$ and the projection block $\pi_{\widehat{\SetK}}$ are depicted in Figures~\ref{fig:dot_product} and~\ref{fig:projection} in detail. FIFO stands for first-in first-out memory and is used to hold the values of the current iterate for use in the next iteration. In the initial iteration, the multiplexers allow $\hat{x}$ and $\hat{\Phi}_n$ through and the result~$\hat{\Phi}_n\hat{x}$ is stored in memory. In the subsequent iterations, the multiplexers allow $\hat{y}_i$ and $I-\hat{H}_n$ through and $\hat{\Phi}_n\hat{x}$ is read from memory.}
		\label{fg-cirucit}
		\end{figure}
	\fi
\fi

\begin{figure}
        \centering
        \begin{subfigure}[t]{0.40\columnwidth}
                \centering
                \includegraphics[width=0.9\textwidth]{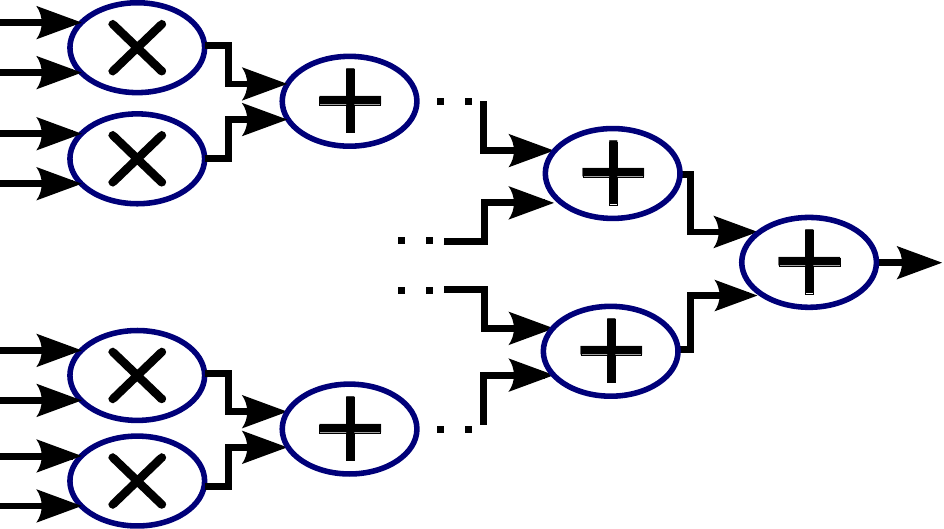}
                \caption{Dot-product block with parallel tree architecture.}
                \label{fig:dot_product}
        \end{subfigure}
	~
        \begin{subfigure}[t]{0.55\columnwidth}
                \centering
                \raisebox{3ex}{
                	\includegraphics[width=0.60\textwidth]{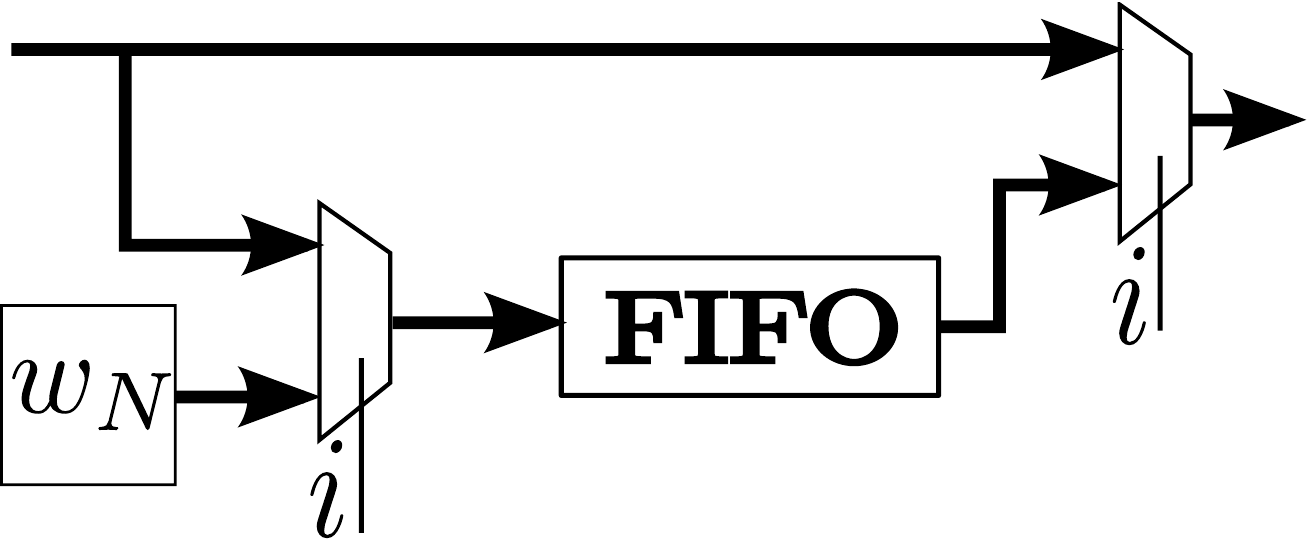}
				}
                \caption{Hardware support for warm-starting, which adds one cycle delay. The last entries of the vector are padded with $w_N$, which can be constant or depend on previous values. }
                \label{fig:warm_start}
        \end{subfigure}
        \caption{Architectures of dot-product and warm-starting.}
\end{figure}


Expression~(\ref{latency expression}) suggests that there will be diminishing returns to parallelization -- a consequence of Amdahl's law. However, (\ref{latency expression})~also suggests that if there are enough resources available, the effect of the problem size on increased computational delay is only logarithmic in the worst case. As Moore's law continues to deliver devices with greater transistor densities, the possibility of implementing algorithms in a fully parallel fashion for medium size optimization problems is becoming a reality.


\subsection{Hardware Architecture for ADMM}

Algorithm~\ref{alg:admm-algorithm} shares the same computational patterns with Algorithm~\ref{alg:fg2-algorithm}. Matrices~$\hat{M}_{11}$ and~$\hat{M}_{12}$ have the same dense structure as matrices~$I-\hat{H}_n$ and~$\hat{\Phi}_n$, hence the high-level architecture is very similar and we do not include it here to avoid replication. The differences lie in the size of the matrices, which affect the number of clock cycles to compute one iteration
\begin{equation}\label{latency_expression_ADMM}
L_A:=\left\lceil \frac{n_A}{P} \right\rceil + l_{A}\lceil \log_2 \left (n_A\right) \rceil + l_{M} +6l_{A} + 2 \, ,
\end{equation}
where $n_A:=N(n_u+n_x + |\soft|)+n_x + |\soft|$,
warm-starting support for variables $z$ and $\nu$ (shown in Figure~\ref{fig:warm_start}), and the projection block for supporting soft state constraints described in Figure~\ref{fig:cone_projection}. This block performs the projection of the pair $(x,\delta)$ onto the set satisfying $\left \{|x-c| \leq r + \delta, \delta \geq 0 \right \}$ by using an explicit solution map for the projection operation and computing the search procedure efficiently. In fact, only $l_A$ extra cycles are needed compared to the standard hard-constrained projection. The block performs a set of comparisons that are used to drive the select signal of a multiplexer.

\ifTwoColumn
\else
	\begin{figure}
        	\centering
        	\begin{subfigure}[t]{0.35\columnwidth}
                	\centering
                	\includegraphics[width=0.95\textwidth]{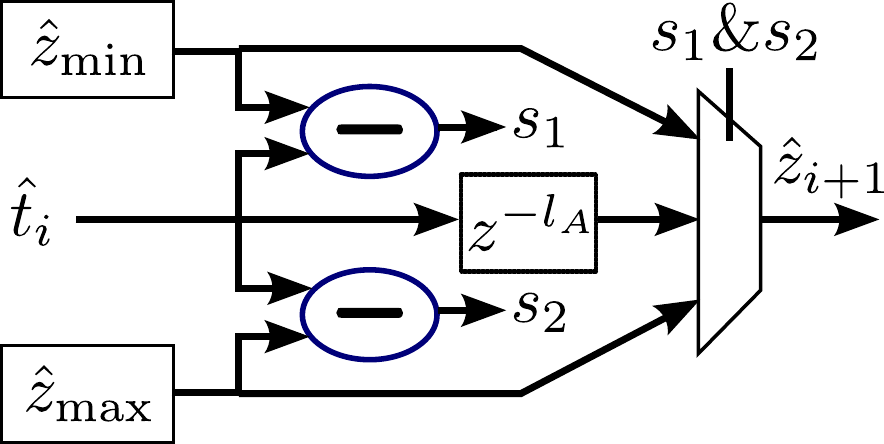}
                	\caption{Box projection block.  The total delay from $\hat{t}_i$ to $\hat{z}_{i+1}$ is $l_A + 1$.}
                	\label{fig:projection}
        	\end{subfigure}
		~
        	\begin{subfigure}[t]{0.625\columnwidth}
                	\centering
                	\includegraphics[width=0.98\textwidth]{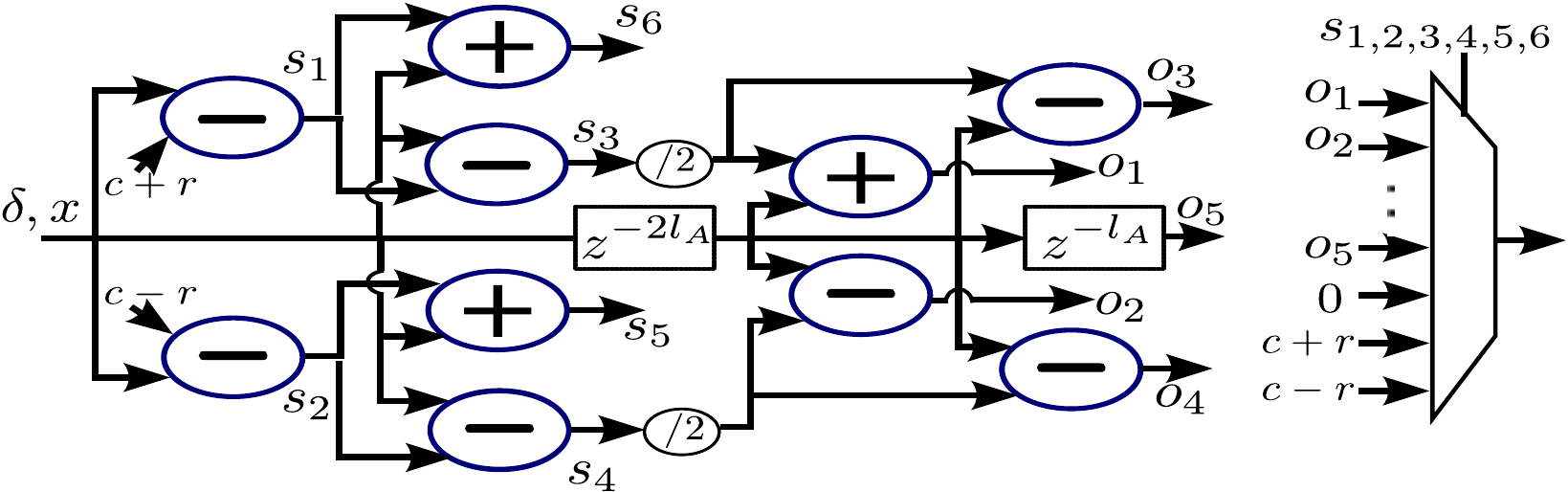}
                	\caption{Cone projection block. The total delay for each component is $2l_A + 1$. $x$ and $\delta$ are assumed to arive and leave in sequence.}
                	\label{fig:cone_projection}
        	\end{subfigure}
        	\caption{Projection architectures. A delay of $l_A$ cycles is denoted by $z^{-l_{A}}$.}
	\end{figure}
\fi

Note that since multiplication and division by powers of two requires no resources in hardware (just a reinterpretation of an array of signals), if $\rho$ is restricted to be a power of two, no hardware multipliers are required in ADMM outside of the matrix-vector multiplication block. Table~\ref{table_resources} compares the resources required to implement the two architectures. Again, with ADMM we trade higher resource requirements and longer delays for the capability to solve more general problems. 

\begin{table}
\caption{Resources required for the fast gradient and ADMM computing architectures.}
\centering
\begin{tabular}{c || cc} 
				& Fast gradient 					& ADMM 				\\ \hline 
 multipliers    			& $P\left[N n_u + 2\right]$ 			& $Pn_A$ 	\\
 adders/subtracters		& $P\left[N n_u + 3\right]$			& $P\left[n_A + 15 \right]$	\\
 memory blocks		& $P\left[N n_u + n_x + 4\right]$		& $P\left[ n_A + 8 \right]$		\\
 size of memory blocks   	& $\left\lceil \frac{N n_u}{P}\right\rceil$	& $\left\lceil \frac{n_A}{P}\right\rceil$		\\
\end{tabular}
\label{table_resources}
\end{table}

Note that in a custom hardware implementation of either of our two methods, the number of execution cycles per iteration is exact. We also employ a fixed number of iterations in our implementations of both algorithms, rather than implementing a numerical convergence test, since such convergence tests represent a somewhat self-defeating computational bottleneck in a hard real-time context. Providing cycle accurate completion guarantees is critical for reliability in high-speed real-time applications~\cite{LeeSeshia2011}.


\section{Numerical Benchmark Study}\label{sec:benchmark}

We reported an implementation of the fast gradient architecture in the preliminary publication~\cite{Jerez2013} to implement an input-constrained MPC controller for a real-world, highly dynamic positioning system inside an atomic force microscope requiring a sampling rate in excess of 1MHz. In this paper, for easier comparison with the existing literature, we use a widely studied benchmark example consisting of a set of oscillating masses attached to walls~\cite{Wang2010,Kogel2011}, as illustrated by Figure~\ref{fig:spring_mass}.  The system is sampled every 0.5 seconds assuming a zero-order hold and the masses and the spring constants have a value of 1kg and 1Nm$^{-1}$, respectively\footnote{Note that we choose this sampling time and parameter set for ease of comparison to other published results. Our implemented methods require computation times on the order of 1$\mu$s, as we report later in this section.}. The system has four control inputs and two states for each mass, its position and velocity, for a total of eight states. The goal of the controller, with parameters $N=10$, $Q=I$ and $R=I$, is to track a reference for the position of each mass while satisfying the system limits.


\begin{figure}
	\centering
    	\includegraphics[width=0.5\columnwidth]{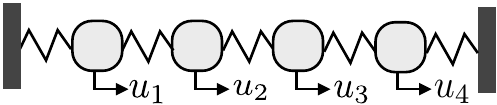}
	\caption{Oscillating masses example.}
	\label{fig:spring_mass}	
\end{figure}

We consider first the case where the control inputs are constrained to the interval~$[-0.5,0.5]$ and the optimization problem~(\ref{costfuncQP}) with 40 optimization variables is solved via the fast gradient method. Secondly, we consider additional hard constraints on the rate of change in the inputs on the interval~$[-0.1,0.1]$ and soft constraints on the states corresponding to the mass positions on the interval~$[-0.5,0.5]$. The remaining states are left unconstrained. The state is augmented to enforce input-rate constraints, and the further inclusion of slack variables increases the dimension of the state vector to $n_x=12$. Note that for problems of this size, MPC control designs based on parametric programming \cite{Bemporad2002,Comaschi2012} are generally not tenable, necessitating  online optimization methods. The resulting problem with 216 optimization variables in the form~(\ref{costfuncQP-ADMM}) is solved via ADMM. The closed-loop trajectories using an MPC controller based on a double precision solver running to optimality are shown in Figure~\ref{fig:trajectories}, where all the constraints become active for a significant portion of the simulation. We do not include any disturbance model in our simulation, although the presence of an exogenous disturbance signal would not lead to infeasibility since the MPC implementation includes only soft-constrained states.   Trajectories arising from closed-loop simulation using a controller based on our fixed-point methods are indistinguishable from those in Figure~\ref{fig:trajectories}, so are excluded for brevity.

\ifTwoColumn
\else
	\begin{figure}
		\centering
	 	\begin{subfigure}[t]{0.5\columnwidth}
		\centering
    		\includegraphics[width=0.9\columnwidth]{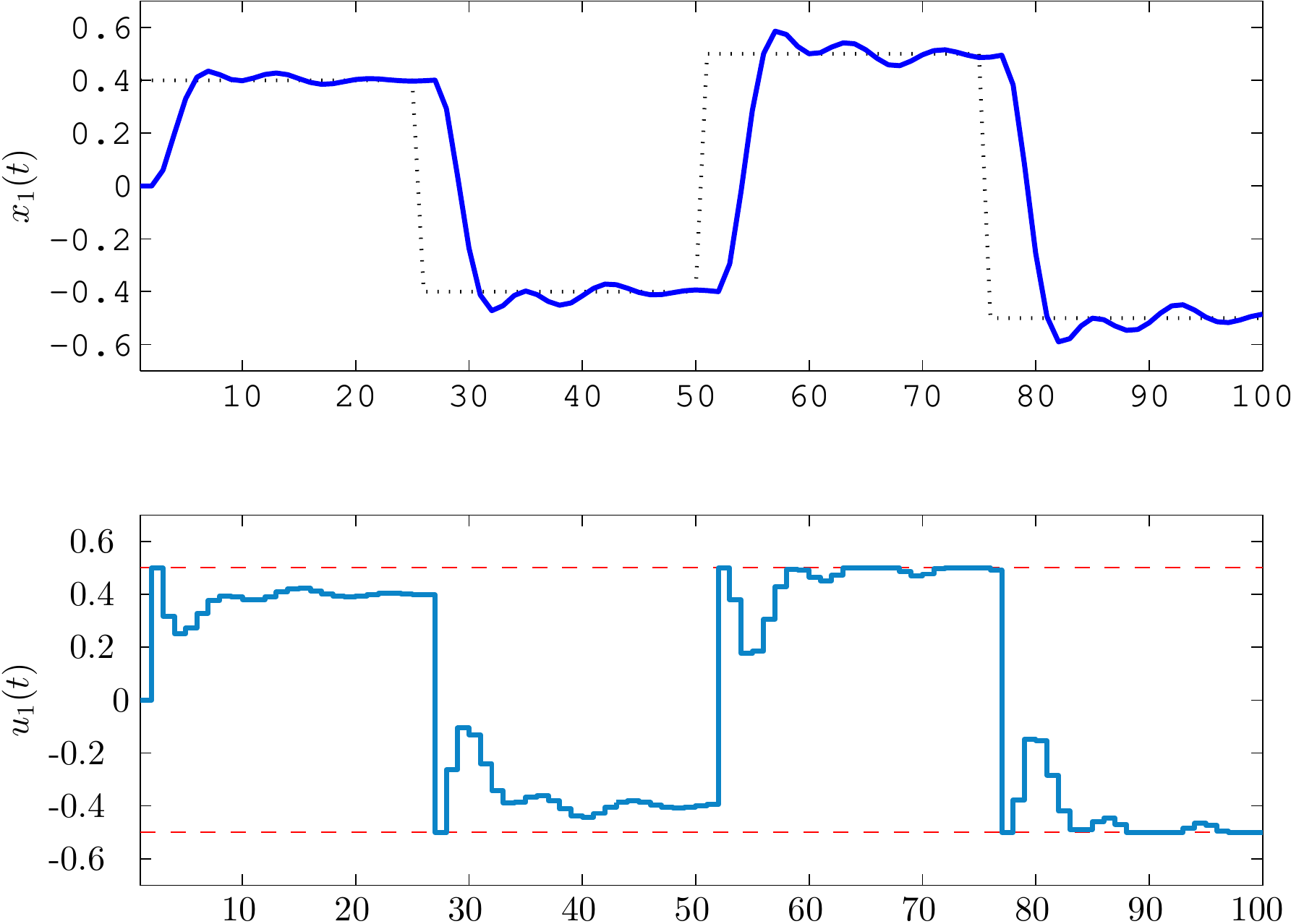}
		\caption{Trajectory with 21 samples hitting the input constraints.}
		\label{fig:input_traj}
        	\end{subfigure}
	 	\begin{subfigure}[t]{0.45\columnwidth}
    		\centering
		\includegraphics[width=\columnwidth]{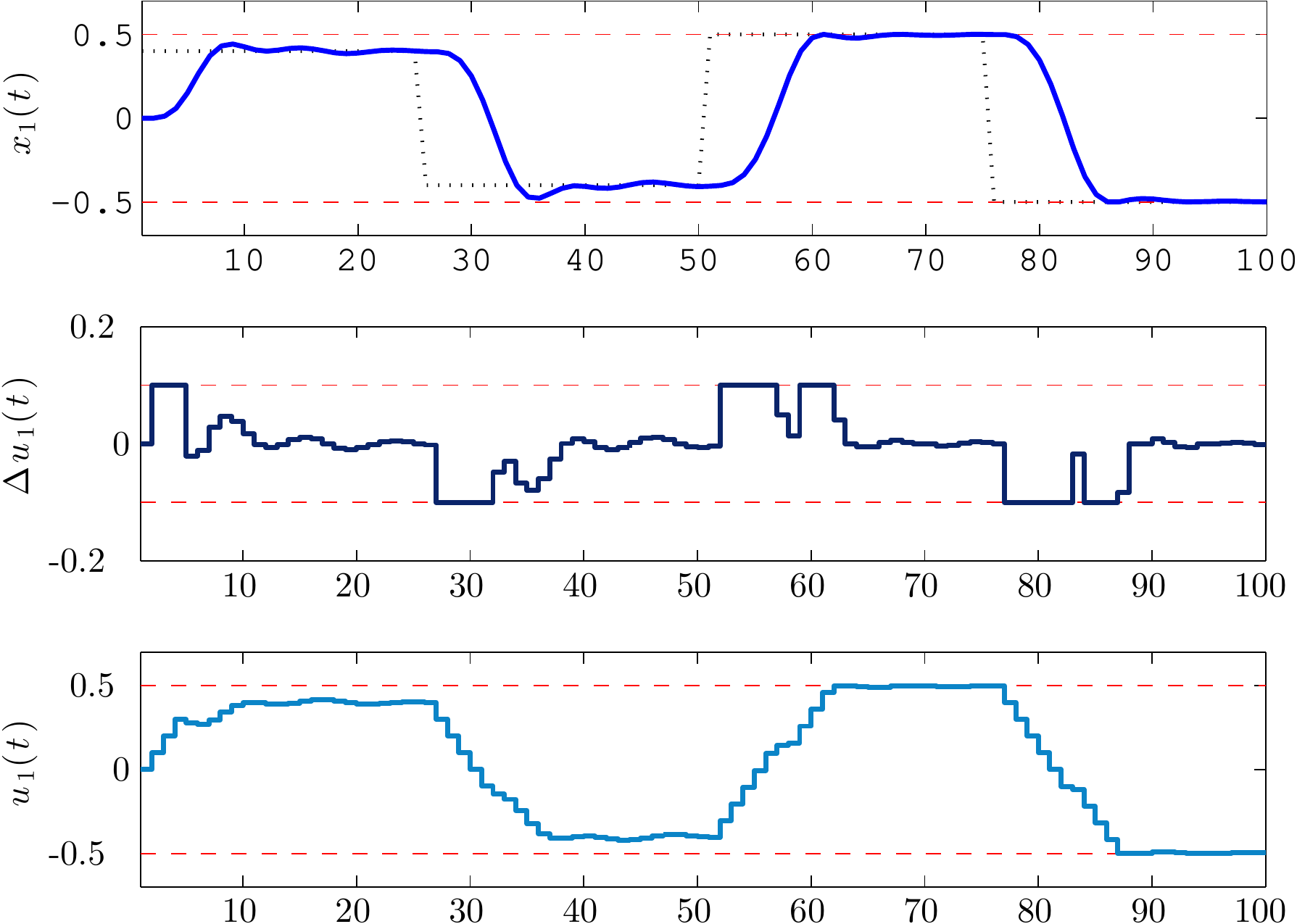}
		\caption{Trajectory with 11, 28 and 14 samples hitting the input, rate and output constraints, respectively.}
		\label{fig:state_traj}
        	\end{subfigure}
		\caption{Closed-loop trajectories showing actuator limits, desirable output limits and a time-varying reference. MPC allows for optimal operation on the constraints.}
		\label{fig:trajectories}
	\end{figure}  
\fi

As a reference for later comparison, an input-constrained problem with two inputs and 10 states, formulated as an optimization problem of the form~(\ref{costfuncQP}) with 40 variables, was solved in~\cite{Kogel2011} using the fast gradient method in approximately 50~$\mu$seconds. In terms of state-constrained implementations, a problem with three inputs and 12 states, formulated as a sparse quadratic program with hard state constraints and 300 variables, was solved in~\cite{Wang2010} using an interior-point method reporting computing times in the region of 5 milliseconds, while the state constraints remained inactive. In both cases, the solvers were implemented in software on high-performance desktop machines.

\ifTwoColumn
\else
	\begin{figure}
		\centering
	 		\begin{subfigure}[t]{0.5\columnwidth}
			\centering
    			\includegraphics[width=0.9\columnwidth]{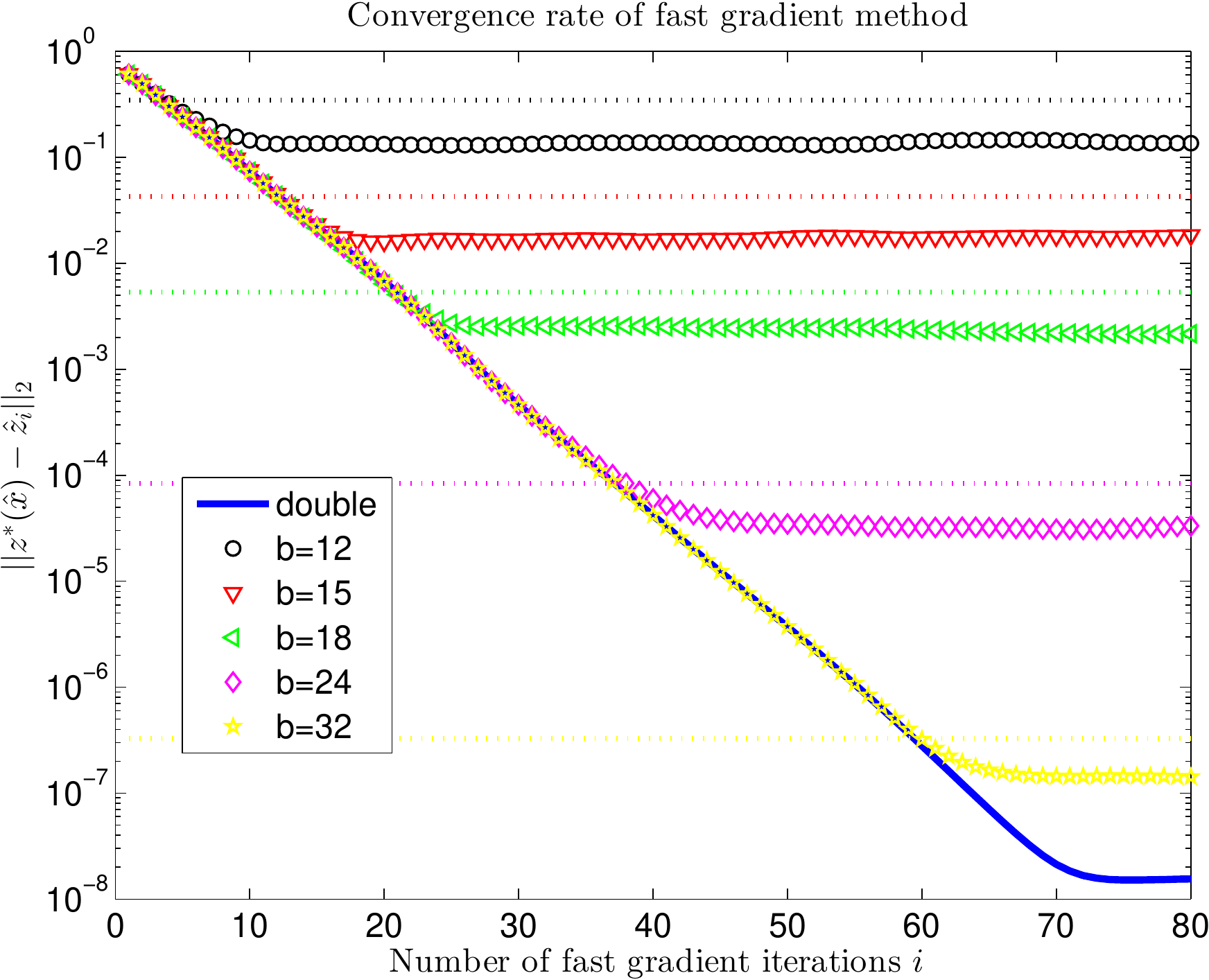}
			\label{fg-convergence}
			\end{subfigure}
	 		\begin{subfigure}[t]{0.45\columnwidth}
    			\centering
			\includegraphics[width=\columnwidth]{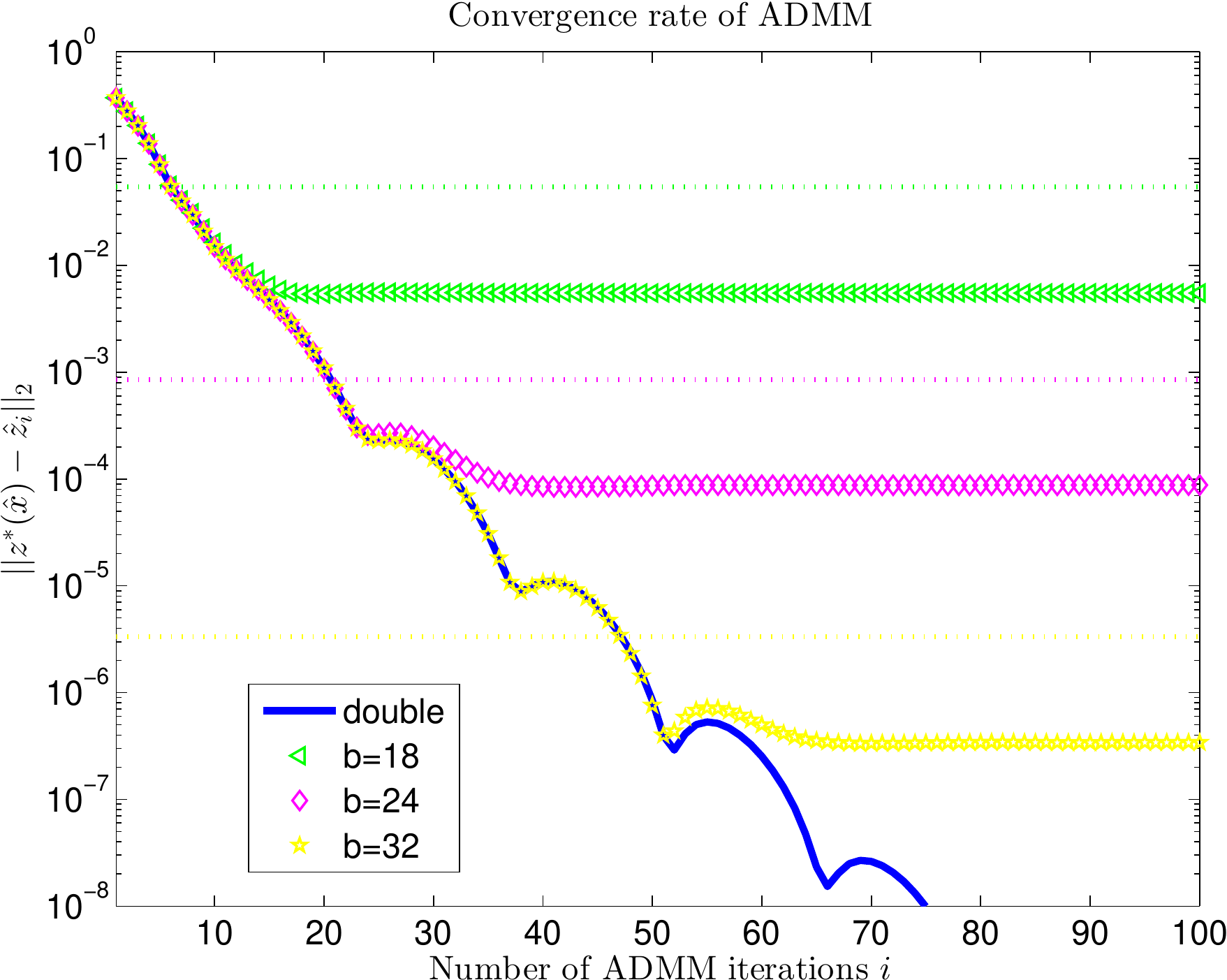}
			\label{fig:admm-convergence}
        		\end{subfigure}
		\caption{Theoretical error bounds given by~(\ref{eq:etabound}) and practical convergence behavior of the fast gradient method (left) and  ADMM (right) under different number representations.}
		\label{fig:convergence}
	\end{figure}
\fi

Our goal is to choose the minimum number of bits and solver iterations such that the closed-loop performance is satisfactory while minimizing the amount of resources needed to achieve certain sampling frequencies. Figure~\ref{fig:convergence} shows the convergence behavior of the fast gradient method and ADMM for two samples in the simulation with an actively constrained solution. The theoretical error bounds on the residual round-off error~$\eta_i$, given by~(\ref{eq:etabound}), allow one to make practical predictions for the actual error for a given number of bits, which, as predicted by Lemma~\ref{lem:etabound} and \eqref{eq:conv1} and \eqref{eq:conv2}, converges to a finite value. Table~\ref{table_closedloop} shows the relative difference in closed-loop tracking performance for different fixed-point fast gradient and ADMM controllers compared to the optimal controller. Assuming that a relative error smaller than $0.05\%$ is desirable, using 15 solver iterations and 16 fraction bits would be a suitable choice for the fast gradient method. The problem~(\ref{costfuncQP-ADMM}) solved via ADMM appears more vulnerable to reduced precision implementation, although satisfactory control performance can still be achieved using a surprisingly small number of bits. In this case, employing more than 18 fraction bits or more than 40 ADMM iterations results in insignificant improvements.

\ifTwoColumn
\else

\begin{table}
\caption{Percentage difference in average closed-loop cost with respect to a standard double precision implementation.  In each table, $b$ is the number of fraction bits employed and $I_{\max}$ is the (fixed) number of algorithm iterations. In certain cases, the error increases with the number of iterations due to increasing accumulation of round-off errors.}
\centering
	 \begin{subtable}[t]{0.48\columnwidth}
\begin{tabular}{c | ccccccc} 
 $I_{\max}\backslash b$ 	&  10		& 12 		& 14 		& 16 		& 18 		& 20 		\\  \hline
 5 				& 5.30	& 2.76	& 2.87	& 3.03	& 3.05	& 3.06	\\
 10 				& 14.53	& 0.14 	& 0.06	& 0.18	& 0.20	& 0.02	\\
 15				& 17.04 	& 0.35 	& 0.25 	& 0.04 	& 0.00 	& 0.01 	\\
 20    				& 16.08 	& 0.15 	& 0.19 	& 0.06 	& 0.01 	& 0.00 	\\
 25    				& 17.27 	& 0.15 	& 0.19 	& 0.05 	& 0.01 	& 0.00 	\\
 30    				& 16.90 	& 0.31 	& 0.21 	& 0.03 	& 0.02 	& 0.00 	\\
 35    				& 18.44 	& 0.19 	& 0.22 	& 0.05 	&  0.01	& 0.00 	\\
\hline 
\end{tabular}
\caption{FGM}
	\end{subtable}
	 \begin{subtable}[t]{0.45\columnwidth}
\begin{tabular}{c | cccccccc} 
 $I_{\max}\backslash b$ 	&  10		& 12 		& 14 		& 16 		& 18 		& 20 		\\  \hline
 10 				& 53.49	& 0.18 	& 1.17	& 0.68	& 0.57	& 0.58	\\
 15				& 47.84 	& 0.46 	& 1.08 	& 0.63 	& 0.51 	& 0.49 	\\
 20    				& 44.79 	& 0.76 	& 0.95 	& 0.57 	& 0.45 	& 0.42 	\\
 25    				& 47.03 	& 0.98 	& 0.86 	& 0.51 	& 0.39 	& 0.37	\\
 30    				& 45.17 	& 1.02 	& 0.82 	& 0.46 	& 0.35	& 0.32 	\\
 35    				& 46.02 	& 1.07 	& 0.81 	& 0.43 	& 0.31	& 0.28 	\\
 40    				& 46.87 	& 1.29 	& 0.74 	& 0.41 	& 0.28	& 0.25 	\\
\hline 
\end{tabular}
\caption{ADMM}
	\end{subtable}
\label{table_closedloop}
\end{table}

\fi


For the implementation of ADMM there are a number of tuning parameters left to the control designer. Setting the regularization parameter to $\rho=2$ simplifies the implementation and provided good convergence behavior. The maximum observed value for the Lagrange multipliers~$\nu$ was~$7.8$, so the penalty parameter $\sigma_1$ was set to~$\sigma_1=8$ to obtain an exact penalty formulation as described by Theorem~\ref{thm:exactpen}. In Section~\ref{const_scaling} it was noted that the convergence of ADMM can be very slow when there is large mismatch between the size of the primal and dual variables. This problem can be largely avoided by scaling the matching condition~(\ref{ADMMmatching}) with a diagonal matrix, where the entries associated with the soft-constrained states and the slack variables are assigned $\sigma$ and the rest are assigned $1$. This scaling procedure correspond to variable transformations $y=D\tilde{y}$ and $z=D\tilde{z}$ that can be applied offline.  

In order to evaluate the potential computing performance the architectures described in Section~\ref{hw-sec} were implemented in FPGAs. For a fixed number of iterations one can calculate the execution time of the solver deterministically according to~(\ref{latency expression}) or~(\ref{latency_expression_ADMM}). The \ac{FPGA} designs can be clocked at more than 400~MHz using chips from Xilinx's high-performance Virtex~6 family or at more than 230~MHz using devices from the low cost and low power Spartan~6 family. Table~\ref{table_performance} shows the achievable sampling times on the two families for different levels of parallelization. The resource usage is stated in terms of the number of embedded multiplier blocks since this is the limiting resource in these designs. For the input-constrained problem solved via the fast gradient method, one can achieve sampling rates beyond 1~MHz with Virtex 6 devices using a modest amount of parallelization. One can also achieve sampling rates in the region of 700~kHz with Spartan~6 devices consuming in the region of 1~W of power. For the state-constrained problem solved via ADMM, since the number of variables is significantly larger, larger devices are needed and longer computational times have to be tolerated. In this case, achievable solution times range from 40kHz to 200kHz for different Virtex 6 devices. 

Note that the fastest performance numbers reported in the literature are in the millisecond region, several orders of magnitude slower than what is achievable using the techniques presented in this paper.



\ifTwoColumn
\else

\begin{table}[h]
\caption{Resource usage and potential performance at 400MHz (Virtex6) and 230MHz (Spartan6) with 15 and 40 solver iterations for FGM (Table \ref{table_performance:FGM}) and ADMM (Table \ref{table_performance:ADMM}), respectively. The suggested chips in the bottom two rows of each table are the smallest with enough embedded multipliers to support the resource requirements of each implementation.}
\centering
	 \begin{subtable}[t]{\columnwidth}
	\centering
\begin{tabular}{c || ccccccccccccccc} 
 $P$ 			& 1 		& 2 		& 3 		& 4 		& 8 		& 16 		& 32 		\\\hline  
 multipliers    		& 42 		& 84 		& 126		& 168 		& 336 		& 672		& 1344	\\
V6 $T_s$ ($\mu$s)	& 1.95	& 1.20	& 0.98	& 0.82	& 0.64	& 0.56	& 0.53	\\
S6 $T_s$ ($\mu$s)	& 3.39	& 2.09	& 1.70	& 1.43	& 1.10	& 0.98	& 0.91	\\
V6 chip 		& LX75	& LX75	& LX75	& LX75 	& LX130	& LX240	& SX315	\\
S6 chip    		& LX45	& LX75	& LX75	& LX100	& -		& -		& -		\\\hline
\end{tabular}
\caption{FGM}\label{table_performance:FGM}
	\end{subtable}
	 \begin{subtable}[t]{\columnwidth}
	\centering
\begin{tabular}{c || ccccccccccccccc} 
 $P$ 			& 1 		& 2 		& 3 		& 4 		& 5 		& 6 		& 7 		\\\hline  
 multipliers    		& 216 		& 432 		& 648		& 864 		& 1080 	& 1296	& 1512	\\
V6 $T_s$ ($\mu$s)	& 23.40	& 12.60	& 9.00	& 7.20 	& 6.20	& 5.40	& 4.90	\\
S6 $T_s$ ($\mu$s)	& 40.70	& 21.91	& 15.65	& 12.52	& 10.78	& 9.39	& 8.52	\\
V6 chip 		& LX75	& LX130	& LX240	& LX550	& SX315	& SX315	& SX475	\\
S6 chip    		& -		& -		& - 		& - 		& - 		& - 		& -		\\\hline
\end{tabular}
\caption{ADMM}\label{table_performance:ADMM}
	\end{subtable}
\label{table_performance}
\end{table}

\fi



\section{Acknowledgements}
This work was supported by the EPSRC (Grants EP/G031576/1 and EP/I012036/1) and the EU FP7 Project EMBOCON, as well as industrial support from Xilinx, the Mathworks, and the European Space Agency. 




\bibliographystyle{./Latex_Class_Files/IEEEtran}
\bibliography{journal2013}

\end{document}